\documentclass[12pt]{article}
\usepackage{amsfonts}
\usepackage[latin1]{inputenc}
\usepackage{amssymb}
\usepackage{amsmath}
\usepackage{mathrsfs}
\usepackage{amsthm}
\usepackage{amscd}
\usepackage{amsbsy}
\usepackage{epsfig}
\usepackage{thumbpdf}

\newtheorem{theorem}{Theorem}
\newtheorem{definition}{Definition}
\newtheorem{lemma}{Lemma}
\newtheorem{proposition}{Proposition}
\newtheorem{corollary}{Corollary}

\newcommand\unit{\hbox{\rm 1\kern-2.8truept l}}

\newcommand\Lform{{\mathcal{L}}\kern-7.56pt\raise1.0pt\hbox{$-$}}
\newcommand{\initsp}{\mathsf{h}}

\newcommand{\tr}{{\mathsf{tr}}}
\newcommand{\flip}{\mathsf{F}}

\begin{document}

%\centerline{\bf  L. Accardi, F. Fagnola, R. Quezada}

\title{Infinite Dimensional Choi-Jamio\l kowski States and Time Reversed Quantum Markov Semigroups}

\author{Jorge R. Bola$\tilde{\textrm{n}}$os-Serv\'in$^{*}$ and Roberto Quezada$^{\dagger}$ \\ \\ 
Universidad Aut\'onoma Metropolitana, Iztapalapa Campus \\ 
Av. San Rafael Atlixco 186, Col. Vicentina\\ 09340 Iztapalapa D.F., Mexico. \\ 
$^*$E-mail: kajito@gmail.com, $^{\dagger}$E-mail: roqb@xanum.uam.mx }

\maketitle

\begin{abstract} We propose a definition of infinite dimensional Choi-Jamio\l kowski state associated with a completely positive trace preserving map. We introduce the notion of $\Theta$-KMS adjoint of a quantum Markov semigroup, which is identified with the time reversed semigroup. The break down of $\Theta$-KMS symmetry (or $\Theta$-standard quantum detailed balance in the sense of Fagnola-Umanit$\grave{\textrm{a}} $\cite{franco2}) is measured by means of the von Neumann relative entropy of the Choi-Jamio\l kowski states associated with the semigroup and its $\Theta$-KMS adjoint.    
\end{abstract}

%\keywords{Infinite dimensional Choi-Jamio\l kowski state, $\Theta$-KMS adjoint, von Neumann relative entropy, entropy production, circulant QMS.}
%\ccode{AMS Subject Classification: 46L55, 82C10, 60J27}

\section{Introduction}

Starting with the work of Agarwal \cite{agarwal}, several notions of quantum detailed ba\-lan\-ce for quantum Markov semigroups (QMS) have been proposed. Roughly speaking, all of these conditions are based on a notion of dual or adjoint. Indeed, given a uniformly continuous QMS ${\mathcal T}=({\mathcal T}_{t})_{t\geq 0}$ with a faithful invariant state $\rho$ and $s\in[0, \frac{1}{2}]$, the $s$-adjoint QMS ${\mathcal T}^{(s)} = ({\mathcal T}^{(s)}_{t})_{t\geq 0}$ is defined by the duality relation 
\begin{eqnarray}
\tr\Big(\rho^{s} x \rho^{1-s} {\mathcal T}_{t}(y)\Big) = \tr\Big(\rho^{s} {\mathcal T}^{(s)}_{t}(x) \rho^{1-s} y\Big). 
\end{eqnarray} It has been proved \cite{franco3}, that among all $s\in[0,\frac{1}{2}]$, there are two prototypical  values: $s=0$ and $s=\frac{1}{2}$. The case $s=\frac{1}{2}$ corresponds with the KMS symmetry discussed by Petz \cite{denes},  Goldstein and Lindsay \cite{gol-lind}, and Cipriani \cite{fabio1,fabio2}, i.e., the QMS ${\mathcal T}$ is KMS symmetric if and only if ${\mathcal T}={\mathcal T}^{(\frac{1}{2})}$. 

While Equilibrium steady states $\rho$ are identified with those satisfying  the above symmetry condition, the break down of this symmetry is identified with the existence of a non-equilibrium steady state $\rho$ and the deviation from equilibrium can be measure by means of a numerical index:  the von Neumann relative entropy of ${\mathcal T}$ and its adjoint QMS, identified with the time reversed semigroup. In fact, the relative entropy of ${\mathcal T}$ and ${\mathcal T}^{(s)}$, $s=0, \frac{1}{2}$, give two different indices that measure deviation from the $s=0$ symmetry (${\mathcal T}={\mathcal T}^{(0)}$) and the deviation from KMS symmetry, respectively. 

Since the von Neumann relative entropy is a function defined on states, the pro\-blem of associating a family of states with a QMS naturally arises. Choi-Jamio\l kowski states associated with trace preserving completely positive maps are well defined and understood in finite dimension. Taking advantage of this correspondence, in our previous work \cite{b-q},  we associated with a QMS and its KMS adjoint the corresponding families of their Choi-Jamio\l kowski states to compute explicitly the relative entropy and entropy production rate for circulant QMS. Apart from the recent work of Holevo \cite{holevo}, infinite dimensional Choi-Jamio\l kowski states have not been studied. Unfortunately Holevo's definition is not the suitable notion in our approach. One of the main aims of this work is to provide an appropriate definition of infinite dimensional Choi-Jamio\l  kowski states as functions of a fixed reference state $\rho$, that we call $\rho$-Choi-Jamio\l kowski states and allow us to define the von Neumann relative entropy and entropy production of a QMS and its adjoint. The reference state $\rho$ helps to control the divergences arising in infinite dimension. We remark that in finite dimension our $\rho$-Choi-Jamio\l kowski states reduces to the usual ones when $\rho$ coincides with the normalized unit. 

The condition of quantum detailed balance discussed by Agarwal includes a typical quantum feature, that of parity of observables under a time reversal operation $\Theta$, namely a linear anti-homomorphism acting on observables $x$ such that $\Theta(xy)=\Theta(y)\Theta(x)$ and being a $*$-map ($\Theta(x^*)=\Theta(x)^*$), that satisfies $\Theta^{2}=I$.  An observable $x$ is even if $\Theta(x)=x$ and it is odd if $\Theta(x)=-x$. As we will see in Theorem \ref{change-of-basis-j-states} below, this reversing operation $\Theta$ and, hence, parity of observables, is intrinsic to the quantum mechanical model; it results from the very fact that the arena of quantum mechanical models is a complex Hilbert space. Also a reversing operation naturally appears in the approach of Accardi-Mohari \cite{acc-mo}, to the study of time reflected Markov processes. Therefore it is natural to incorporate it in the above mentioned symmetry conditions. In this direction Fagnola and Umanit$\grave{\textrm a}$ introduced the notion of $\Theta$-Standard Quantum Detailed Balance ($\Theta$-SQDB), as the $\Theta$-symmetry condition ${\mathcal T}^{(\frac{1}{2})} = \Theta\circ {\mathcal T} \circ\Theta$. Of course another $\Theta$-symmetry condition can be defined using the $s=0$ adjoint ${\mathcal T}^{(0)}$, instead of ${\mathcal T}^{(\frac{1}{2})}$, i.e., ${\mathcal T}^{(0)} = \Theta\circ {\mathcal T} \circ\Theta$. Deviation from these $\Theta$-symmetry conditions can be measured by the von Neumann relative entropy of ${\mathcal T}$ and suitable defined $\Theta$-adjoint semigroups. But no notion of $\Theta$-KMS adjoint can be found in the literature. To fulfill this gap we introduce the $\Theta$-KMS adjoint defined as the unique QMS ${\mathcal T}^{\Theta} = ({\mathcal T}^{\Theta}_{t})_{t\geq 0}$ satisfying the following $\Theta$-KMS duality relation 
\begin{eqnarray}\label{theta-kms-adjoint}
\tr\Big(\rho^{\frac{1}{2}}\Theta(x^*)\rho^{\frac{1}{2}}{\mathcal T}_{t}(y)\Big)= \tr\Big(\rho^{\frac{1}{2}}\Theta({\mathcal T}^{\Theta}_{t}(x^*)) \rho^{\frac{1}{2}} y \Big), \;  \forall \;  x,y \in {\mathcal B}(\initsp).  
\end{eqnarray} We shall use the $\Theta$-KMS adjoint to study deviation from $\Theta$-SQDB, that is con\-si\-de\-red the most natural quantum extension of the classical detailed balance condition, by means of the relative entropy of ${\mathcal T}$ and ${\mathcal T}^{\Theta}$. One could also define the $\Theta$-adjoint for the $s=0$ case and the corresponding relative entropy. 

We shall prove some remarkable properties of the von Neumann relative entropy of a uniformly continuous QMS ${\mathcal T}$ and its $\Theta$-KMS adjoint ${\mathcal T}^{\Theta}$, and study the corresponding notion of entropy production. Our approach gives a general scheme that can be applied to study properties of the relative entropy of any pair of QMS, and measure the deviation from any other of the symmetry conditions mentioned above. Our main observation is that the $\Theta$-KMS adjoint is the most natural quantum extension of the classical time reversal semigroup, since its family of $\rho$-Choi-Jamio\l kowski states provide the densities of the backward states associated with $\Theta$-SQDB. Our approach works well in any separable initial Hilbert space $\initsp$ and it reduces to the approach outlined by Fagnola and Rebolledo \cite{fag-reb}, for finite dimensional $\initsp$.  We stress that our approach allows us to prove that the backward state's density is given by the $\rho$-Choi-Jamio\l kowski states $\mathcal J_\rho (\mathcal T^\Theta_{*t})$, $t\geq 0$, associated with the $\Theta$-KMS adjoint semigroup, for any separable Hilbert space $\initsp$.

In addition to the symmetry conditions mentioned above, there are other well known characterizations of equilibrium steady states: Boltzmann-Gibbs prescription, quantum detailed balance in the sense of Kossakowski, Frigerio, Gorini and Verri \cite{kfgv,fg} and the KMS condition, among others. Modified versions of these equilibrium conditions such as non-linear Boltzmann-Gibbs prescription, weighted detailed balance and local KMS condition, respectively, have been discussed recently in \cite{a-f-q} These modifications allow one to include, beside the equilibrium, non-equilibrium steady states associated with quantum currents describing the flow of energy from the environment to the system.

After some preliminaries presented in Section \ref{preliminaries}, in Section \ref{choi-jamiolkowski-states} we define our infinite dimensional $\rho$-Choi-Jamio\l kowski states and study some of its remarkable properties. Later, in Section \ref{entropy-production}, we prove that the von Neumann relative entropy of the $\rho$-Choi-Jamio\l kowski states associated with ${\mathcal T}$ and ${\mathcal T}^{\Theta}$ is independent of the orthonormal basis used to define these states, define our notion of entropy production rate and deduce an explicit formula to compute it. Finally, in Section \ref{example}, we use our formula to compute the entropy production rate for the class of circulant QMS introduced in our previous work \cite{b-q}.

\section{Preliminaries}\label{preliminaries}

By $\initsp$ with denote a separable Hilbert space endowed with an inner product $\langle \cdot, \cdot \rangle$. As usual, the von Neumann algebra of all bounded operators on $\initsp$ will be denoted by $\mathcal B(\initsp)$, while the Banach space of all finite trace operators, endowed with the trace norm $\|\eta\|_{1}=\tr |\eta|$, will be denote by $L_1(\initsp)$.  Along this work we use an anti-unitary operator $\theta$, for the sake of completeness we recall its definition and some of its properties.
\subsection{Anti-unitary Operators}
\begin{definition} A bijective, anti-linear operator $\theta:\initsp \longrightarrow \initsp $ is called anti-unitary if
\begin{align*}\langle \theta u, \theta v \rangle=\langle v,u\rangle, \ \ \ \text{ for all } x,y \in \initsp .\end{align*} 
\end{definition}

It is immediate from the definition that anti-unitary operators are bounded operators. Even more, they are antilineal isometries, and so they send orthonormal bases on orthonormal bases. The most used anti-unitary operators in physics are those satisfying $\theta^2=\unit$, a property that we assume from now on. 

The following properties are straightforward.

\begin{proposition}An anti-unitary operator $\theta$ has the following properties:
\begin{enumerate}
\item[(i)] Its adjoint $\theta^*$ is also antilinear and it is defined by $\langle u, \theta v \rangle = \langle v, \theta^* u \rangle.$ If $\theta^2= \unit$, then $\theta=\theta^*$.
\item[(ii)] $\theta \theta^* =\theta^* \theta =\unit$.
\item[(iii)] $\theta x \theta$ is a linear operator satisfying $(\theta x \theta)^*= \theta^* x^* \theta^*$. If $\theta^2=\unit$, then $(\theta x \theta)^*=\theta x^*\theta $.
\item[(iv)] The composition of two anti-unitary operators is unitary.
\item[(v)]The composition of a unitary and an anti-unitary operator is an anti-unitary operator.
\item[(vi)] Each anti-unitary operator $\theta$ is the composition of an unitary and a conjugation w.r.t.  an  orthonormal basis.
\end{enumerate} \end{proposition}

Due to $(vi)$, when dealing with an anti-unitary operator and an orthonormal basis $\{e_i\}_i$, up to a unitary transformation we can identify $\theta$ with the conjugation w.r.t. $\{e_i\}_i$. So that $\theta e_{i} = e_{i}$ and for $u=\sum_{i}u_{i} e_{i}$, $\theta u=\sum_{i} \bar{u}_{i} e_{i}$.

As we have seen, one needs to be careful when operating with anti-unitary operators, since their behaviour can be rather  counter-intuitive.

\begin{proposition}\label{theta-proy}Let $\theta$ be the anti-unitary operator of conjugation w.r.t. the orthonormal basis $\{e_i\}_i$. The following properties hold:
\begin{itemize}  
\item[(i)] $\theta | e_i \rangle \langle e_j | \theta = |\theta e_i \rangle \langle \theta e_j |=| e_i \rangle \langle e_j |$ 
\item[(ii)] $\theta | e_i \rangle \langle e_j |=| e_i \rangle \langle e_j | \theta$ \end{itemize}  \end{proposition}

\subsection{Standard quantum detailed balance with a reversing operation}
Associated with an anti-unitary operator $\theta$ is a reversing operation on the ob\-ser\-va\-bles: $\Theta(x)=\theta x^* \theta$. This reversing operation allows us to incorporate in the KMS symmetry, typical quantum notions such as that of parity of observables. From now on the KMS adjoint QMS ${\mathcal T}^{(\frac{1}{2})}$ will be denoted simply by ${\mathcal T}'$.
\begin{definition}\label{bal-det-theta} A uniformly continuous QMS $(\mathcal T)_{t\geq0}$  with a faithful invariant state $\rho$ and a KMS adjoint semigroup $(\mathcal T^\prime)_{t\geq0}$, generated by ${\mathcal L}'$,  satisfies a Standard Quantum Detailed balance condition with respect to the reversing operation $\Theta$ ($\Theta$-SQDB) if
\begin{align*} \mathcal T_t^\prime=\Theta \circ \mathcal T_t \circ \Theta. \end{align*}
\end{definition} 

\subsection{Circulant quantum Markov semigroups}\label{circ-qms}
Circulant QMS where introduced in our previous work \cite{b-q}. We recall here the main properties of this class of semigroups as well as some well known results about circulant matrices.

A circulant matrix $Q=(q_{ij})_{0\leq i,j\leq p-1}$ with $p\in{\mathbb N}$ is a $p\times p$ complex matrix satisfying $q_{ij}=\alpha(j-i \; \textrm{mod} \, p)$, for some vector $\alpha\in{\mathbb C}^{p}$. The primary permutation matrix $J_p$ is defined as $J_p =\sum_{i} |e_{i}\rangle \langle e_{i+1}|$, where all index numbers must be considered as elements in the group ${\mathbb Z}_{p}$. 
 
\begin{definition} \label{circulant-cp-map} 
\begin{itemize}
\item[(i)] Let $J_p$ be the primary permutation matrix. The CP linear map defined on the space of $p\times p$ complex matrices $\mathcal M_p (\mathbb C)$ by \[\Phi_{*} (x)=\sum_{i=0}^{p-1} \alpha(p-i) J^{i}_p x J^{*i}_p,\] for some $\alpha(i)\geq 0$ is called circulant CP map. 

\item[(ii)] Let $J_{p},J_{q}$ be the primary matrices for $p,q\in{\mathbb N}$. The CP linear map defined on $\mathcal M_p (\mathbb C)\otimes \mathcal M_q (\mathbb C)$ by \[\Phi_{*} (x)=\sum_{i=0,j=0}^{p-1,q-1} \alpha(p-i,q-j) (J^{i}_{p}\otimes J^{j}_{q})  x (J^{i}_{p}\otimes J^{j}_{q})^{*},\] for some $\alpha(i,j)\geq 0$ is called block circulant CP map. 
\end{itemize}

We will call simply \textbf{circulant} any one of these CP maps. 
\end{definition}

Consider the discrete time Markov chain on the abelian group $\mathbb Z_p$ (respectively  ${\mathbb Z}_{p}\times {\mathbb Z}_{q}$) associated with a given probability distribution $\alpha:{\mathbb Z}_{p}\mapsto [0,1]$ $\big(\alpha:{\mathbb Z}_{p}\times {\mathbb Z}_{q}\mapsto [0,1]\big)$ with $\sum_{i}\alpha(i)=1$ $(\sum_{i,j}\alpha(i,j)=1)$. If we set $\alpha(0)=0$ $\big(\alpha(0,0)=0\big)$, then the corresponding \textit{bi-stochastic} transition probabilities matrix \begin{align*}&\Pi =\sum_{i}\alpha(i) (J_{p}^{i}),\\ \Big(&\textrm{respectively}, \; \Pi =\sum_{i,j}\alpha(i,j) (J_{p}^{i}\otimes J_{q}^{j})\Big),\end{align*} with $J_{p}$ the primary permutation matrix, can be considered as the transition probability  matrix of the embedded Markov chain of the continuous time Markov chain with infinitesimal generator $Q=\Pi-\unit$, where $\unit$ denotes the identity matrix in ${\mathcal M}_{p}(\mathbb C)$ (respectively ${\mathcal M}_{p}(\mathbb C)\otimes {\mathcal M}_{q}(\mathbb C)$). Clearly, $Q$ is a circulant (respectively, block circulant with circulant blocks) matrix. We shall consider the quantum extensions, in pre-dual representation, 
\begin{align}\label{cycle-rep-phi} 
&\Phi_{*}(x)=\sum_{i\in{\mathbb Z}_p}\alpha(p-i) J_{p}^{i}x J_{p}^{*i}\\ \Big(&\Phi_{*}(x)=\sum_{(i,j)\in{\mathbb Z}_{p}\times{\mathbb Z}_{q}}\alpha(p-i,q-j) (J_{p}^{i}\otimes J_{q}^{j})x(J_{p}^{i}\otimes J_{q}^{j})^{*}\Big)\nonumber , 
\end{align} and 

\begin{eqnarray}\label{cycle-rep-gen}
{\mathcal L}_{*}(x)=\Phi_{*}(x)-x, 
\end{eqnarray} of $\Pi$ and $Q$, respectively, with $x\in {\mathcal M}_{p}(\mathbb C)$ $\Big(x\in{\mathcal M}_{p}(\mathbb C)\otimes {\mathcal M}_{q}(\mathbb C)\Big)$. According to Definition  \ref{circulant-cp-map}, $\Phi_{*}$ is a circulant CP map. We call ${\mathcal L}_{*}$ a circulant GKSL generator and circulant QMS the semigroup generated by ${\mathcal L}_{*}$. 

The set of circulant matrices is an abelian sub-algebra of ${\mathcal M}_{p}(\mathbb C)$, whose e\-le\-ments are simultaneously diagonalized by the discrete Fourier transform $F_p =\frac{1}{\sqrt{p}}\sum_{0\leq k,l\leq p-1}\omega^{kl} |e_k \rangle \langle e_l |$, where $\omega$ is a primitive $p$-root of identity. Indeed we have the following result proven in \cite{b-q} 

\begin{theorem}\label{prop-circ-matrix}
If $Q=\sum_{i,j} \alpha(i,j) J_{p}^{i}\otimes J_{q}^{j},$ then 
\begin{itemize}  
\item[(i)]  \[(F_{p}\otimes F_{q})Q(F_{p}\otimes F_{q})^{*}=\sum_{k,l} \lambda_{k,l} |e_{k}\otimes e_{l}\rangle\langle e_{k}\otimes e_{l}|,\] with $\lambda_{k l}=\sum_{i,j}\alpha(i,j)\overline{\omega}_{p}^{ik}\overline{\omega}_{q}^{jl}$, and 
 
\item[(ii)] \[e^{t Q}= \frac{1}{pq}\sum_{i,j,m,n}\Phi_{m-i, n-j}(t)|e_{i}\otimes e_{j}\rangle \langle e_{m}\otimes e_{n}|,\] with $\Phi_{i,j}(t)=\sum_{k,l}\omega_{p}^{ik}\omega_{q}^{jl} e^{t\lambda_{kl}}$. 

\end{itemize}
\end{theorem}

The above properties of circulant matrices reflect in corresponding properties of circulant QMS, also proven in \cite{b-q} 

\begin{theorem}\label{prop-circ-qms}
The semigroup $\mathcal T_*$ generated by ${\mathcal L}_{*}(x)=\Phi_{*}(x)-x$, with $\Phi_{*}(x)=\sum_{(i,j)\in{\mathbb Z}_{p}\times{\mathbb Z}_{q}}\alpha(p-i,q-j) (J_{p}^{i}\otimes J_{q}^{j})x(J_{p}^{i}\otimes J_{q}^{j})^{*}$ satisfies the following  properties:
\begin{itemize}

\item[(i)] The explicit action of the semigroup is given by 
\begin{align*}\mathcal T_{t}(x)=\frac{1}{pq}\sum_{m,n} \Phi_{m,n}(t) (J^m_p \otimes J^n_q) x (J^m_p \otimes J^n_q)^*,\end{align*} where $\Phi_{m,n}(t)$ are the matrix elements of $e^{tQ}$ in Theorem \ref{prop-circ-matrix}.  

\item[(ii)] The $\rho$-Choi-Jamio\l kowski state associated with $\mathcal T_{*t}$, the invariant $\rho=\frac{1}{pq} \unit$ and the cannonical basis $\{ e_i \otimes e_j \}_{i,j}$ of $\mathbb C^p \otimes \mathbb C^q$ is given by

\begin{align*}
\mathcal J_\rho(\mathcal T_{*t})= \frac{1}{pq} \sum_{m,n}\Phi_{m,n}(t) |u_{mn}\rangle\langle u_{mn}|,
\end{align*} where $u_{mn}= \frac{1}{\sqrt{pq}} \displaystyle\sum_{ij} (e_{i}\otimes e_{j})\otimes(e_{m+i}\otimes e_{n+j})$.
\end{itemize}
\end{theorem}

\section{Infinite dimensional Choi-Jamio\l kowski states}\label{choi-jamiolkowski-states}
Given a state $\rho$ and a fixed orthornormal basis $\{e_i\}_i$ of $\initsp$, let $\omega_\rho:\textrm{span}\{u\otimes v: u, v \in\initsp \} \subset \initsp \otimes \initsp \longrightarrow \initsp \otimes \initsp$ be the linear, possibly unbounded, operator defined on simple tensors $u\otimes v$, $u, v\in \initsp$ by

\begin{eqnarray}\label{def-omega-rho-simple}
\begin{aligned}
\omega_\rho u\otimes v &=\sum_{i,j} \langle e_j \otimes \rho^\frac{1}{2} e_j , u \otimes v\rangle e_i \otimes \rho^\frac{1}{2} e_i.
\end{aligned}
\end{eqnarray} and extended to $\text{span}\{u \otimes v : u,v\in\initsp \}$ by linearity. Since 

 \begin{align*}||\omega_\rho u \otimes v||^2&=\Big \langle\omega_\rho u \otimes v,\omega_\rho u \otimes v \Big \rangle =\sum_{i,j}|\langle \rho^\frac{1}{2} e_k, e_l \rangle|^2 \langle e_j,e_i \rangle \langle \rho^\frac{1}{2} e_j , \rho^\frac{1}{2}e_i \rangle \\
&=\sum_i |\langle \rho^\frac{1}{2} u, v \rangle|^2 \langle e_i , \rho e_i \rangle =|\langle \rho^\frac{1}{2} u, v \rangle|^2 <\infty,\end{align*}  we see that $\omega_\rho u \otimes v$ is an element of $\initsp \otimes \initsp$ for any $u, v\in \initsp$. Hence $\omega_{\rho}$ is well defined on $\textrm{span}\{u\otimes v : u, v \in\initsp\}$. In particular, its is well defined on $\textrm{span}\{e_{k}\otimes e_{l}\}_{k,l}$. We will prove that $\omega_{\rho}$ is a state.

\begin{lemma} The following statements for $\omega_\rho$ hold true:
\begin{itemize}
\item[(i)]  The operator $\omega_\rho$ is bounded on $\textrm{span}\{e_{k}\otimes e_{l}\}_{k,l}$ and it can be continuously extended to the whole space $\initsp \otimes \initsp$. We use the same symbol $\omega_{\rho}$ to denote this extension.
\item[(ii)]$\omega_\rho$ is a positive operator on $\initsp \otimes \initsp$.
\item[(iii)] $\omega_\rho$ is a state.
\end{itemize}
\end{lemma}

\begin{proof}
Observe that if $u\in\text{span}\{e_k\otimes e_l\}$, $u=\sum_{l,k}u_{lk} e_k \otimes e_l$, then 
\begin{align*}\|\omega_\rho u\|^2&=\sum_{k,l,k^\prime,l^\prime}\overline{u}_{kl} u_{k^\prime l^\prime} \langle \omega_\rho e_k \otimes e_l, \omega_{\rho} e_{k^\prime} \otimes e_{l^\prime} \rangle =\sum_{k,l,k^\prime,l^\prime}\overline{u}_{kl} u_{k^\prime l^\prime}\overline{\langle \rho^\frac{1}{2}e_k,e_l\rangle} \langle \rho^\frac{1}{2}e_{k^\prime},e_{l^\prime}\rangle  \\
&=\left|\sum_{k,l} u_{kl} \langle \rho^\frac{1}{2} e_k,e_l \rangle \right|^2 \leq \Big(\sum_{k,l} |u_{kl}|^2\Big) \Big(\sum_{k,l}|\langle \rho^\frac{1}{2} e_k,e_l\rangle |^2\Big)\leq \|u\|^2 \text{tr}\rho.
\end{align*} Therefore $\omega_\rho$ is bounded on $\textrm{span}\{e_{k}\otimes e_{l}\}_{k,l}$ and, by the density of $\textrm{span}\{e_{k}\otimes e_{l}\}_{k,l}$ in $\initsp\otimes\initsp$, it can be continuously extended to the whole space $\initsp \otimes \initsp$. For simplicity of notation, we use the same symbol to denote this extension. This proves $(i)$. 

To prove $(ii)$, take any $u\in \text{span}\{e_k \otimes e_l\}$, then
\begin{align*}\langle u,\omega_\rho u\rangle &= \sum_{k,l,k^\prime,l^\prime} \overline{u}_{kl} u_{k^\prime l^\prime} \langle e_k \otimes e_l , \omega_\rho e_{k^\prime}\otimes e_{l^\prime}  \rangle\\
& =\sum_{\substack{k,l,k^\prime,l^\prime\\i,j}}  \overline{u}_{kl} u_{k^\prime l^\prime}  \delta_{k^\prime j} \delta_{ki}\langle e_{l^\prime},\rho^{\frac{1}{2}} e_j \rangle \langle e_l,\rho^{\frac{1}{2}}e_i\rangle \\
&=\Big| \sum_{k,l} u_{kl} \langle e_l,\rho^{\frac{1}{2}} e_k \rangle \Big|^2 \geq 0.
\end{align*} By a density argument, we get $\langle v ,\omega_\rho v \rangle \geq 0$ for any $v \in \initsp \otimes \initsp$. 

It remains to prove that $\tr({\omega_\rho})=1$. Observe that, 
 \begin{align*} \tr(\omega_\rho)&=\displaystyle\sum_{k,l} \langle e_k \otimes e_l , \omega_\rho e_k \otimes e_l \rangle= \sum_{k,l,i,j}\langle e_k \otimes e_l, |e_i \otimes \rho^{\frac{1}{2}}e_i \rangle \langle e_j \otimes \rho^{\frac{1}{2}} e_j | e_k \otimes e_l \rangle \\
&=\sum_{k,l,i,j} \delta_{jk} \delta_{ki}\langle e_l, \rho^{\frac{1}{2}} e_i \rangle  \langle \rho^{\frac{1}{2}} e_j, e_l \rangle 
=\sum_l \langle\rho^\frac{1}{2}e_l,\rho^\frac{1}{2}e_l \rangle = \tr(\rho)=1, \end{align*}  where Parseval's identity was used in the last equalities. This finishes the proof.
\end{proof}

Let $\Phi_*$ be a bounded CP operator on $L_1(\initsp)$. Denote by ${\mathcal J}_{\rho}(\Phi_{*})$ the sesquilinear form  defined on the span of the orthonormal basis $\{e_k \otimes e_l\}_{kl}$ by means of 
\begin{align*}
{\mathcal J}_{\rho}(\Phi_{*})(u, v)  := \sum_{k,l,k',l'} \overline{u}_{kl} v_{k'l'}  \big\langle e_{l}, \Phi_{*}(|\rho^{\frac{1}{2}} e_k \rangle \langle \rho^{\frac{1}{2}} e_{k'} |)  e_{l'}\big\rangle  \end{align*}  if $u=\sum_{k,l} u_{kl} e_{k}\otimes e_{l}$ and $v=\sum_{k',l'} v_{k'l'} e_{k'}\otimes e_{l'}$.

\begin{lemma}\label{rho-jamio} The associated quadratic form ${\mathcal J}_{\rho}(\Phi_{*})(u,u)$ is bounded and positive on $ \text{span}\{e_k \otimes e_l\}_{kl}$. Hence, the sesquilinear form ${\mathcal J}_{\rho}(\Phi_{*})$ has a bounded extension to the whole $\initsp\otimes \initsp$. 
\end{lemma}

\begin{proof}
For any element in the span of the orthonormal basis $\{e_k \otimes e_l\}_{kl}$, $u=\sum_{k,l} u_{kl} e_k \otimes e_l$, we have 
\begin{eqnarray}\label{sesq-pos}
\begin{aligned}
{\mathcal J}_{\rho}(\Phi_{*}) (u, u) 
&=\sum_{k,l,k',l'} \overline{u}_{kl} u_{k'l'}  \big\langle e_{l}, \Phi_{*}(|\rho^{\frac{1}{2}} e_k \rangle \langle \rho^{\frac{1}{2}} e_{k'} |)  e_{l'}\big\rangle  \\ 
& = \sum_{n} \sum_{k,l,k',l'} \overline{u}_{kl} u_{k'l'}  \big\langle e_{l}, L_{n} |\rho^{\frac{1}{2}} e_k \rangle\langle \rho^{\frac{1}{2}} e_{k'} | L_{n}^{*}  e_{l'}\big\rangle \\ 
&= \sum_{n} \Big|\sum_{k,l} u_{kl} \langle e_{k}, \rho^{\frac{1}{2}} L_{n}^{*} e_{l}\rangle \Big|^{2}. 
\end{aligned} \end{eqnarray}

This shows the positivity of the quadratic form. Now, an application of Cauchy-Schwartz inequality in the last term of (\ref{sesq-pos}), the complete positivity of $\Phi_{*}$ and the fact that, due to the complete positivity of $\Phi_*$,  $||\Phi_*(\sigma)||_1=\text{tr}\Big( \sqrt{\Phi_*(\sigma)^* \Phi_*(\sigma)}\Big)$ $=\text{tr}\Big( \Phi_*(\sigma)\Big)$ if $\sigma$ is positive, yield 
\begin{eqnarray}
\begin{aligned}
|{\mathcal J}_{\rho}(\Phi_{*}) (u, u)| & \leq \sum_{n} \Big(\sum_{kl}|u_{kl}|^2 \Big) \Big(\sum_{kl}|\langle e_k , \rho^{\frac{1}{2}}L_{n}^{*} e_l \rangle|^{2}\Big)\\
&= \| u\|^{2} \sum_{n} \sum_{kl} |\langle e_{k}, \rho^{\frac{1}{2}} L_{n}^{*} e_{l}\rangle|^2 \\
&= \| u\|^{2} \sum_{kl} \sum_{n} \langle e_{k}, \rho^{\frac{1}{2}} L_{n}^{*} e_{l}\rangle  \langle \rho^{\frac{1}{2}} L_{n}^{*} e_{l}, e_k \rangle  \\ 
&=\|u\|^2  \sum_{kl} \Big\langle e_{l}, \Phi_{*}\big(\rho^{\frac{1}{2}}|e_k \rangle \langle e_k| \rho^{\frac{1}{2}}\big) e_l \Big\rangle\\
& = \|u\|^2 \sum_{l} \Big\langle e_{l}, \Phi_{*}\big(\sum_{k}\rho^\frac{1}{2}|e_k \rangle \langle e_k|\rho^\frac{1}{2} \big) e_l \Big\rangle \\ 
&\leq \|u\|^2 \|\Phi_*\|_{\mathcal B(L_1(\initsp))}\left\|\sum_k \rho^\frac{1}{2} |e_k \rangle \langle e_k| \rho^\frac{1}{2}\right\|_1\leq\|u\|^2 \ \|\Phi_*\|_{\mathcal B(L_{1}(\initsp))}.
\end{aligned}
\end{eqnarray}  We have used that the operator $\sigma= \sum_{k}\rho^\frac{1}{2}|e_k \rangle \langle e_k|\rho^\frac{1}{2}$ is clearly positive and has a finite trace; indeed, an application of Parceval's identity yields \[\tr(\sigma)=\sum_{j,k}|\langle e_{j} ,\rho^ {\frac{1}{2}} e_{k}  \rangle|^{2} = \sum_{k}\|\rho^{\frac{1}{2}} e_k \|^{2} = \sum_{k}\langle e_k , \rho e_{k}\rangle = 1.\] 
This proves that the associated quadratic form and hence, by the Cauchy-Schwartz inequality for sesquilinear forms, the sesquilinear form ${\mathcal J}_{\rho}(\Phi_{*})$ is bounded on the span of the orthonormal basis, i.e., 
 \begin{align}\label{continuity-of-J-01}
 |{\mathcal J}_{\rho}(\Phi_{*})(u,v)|\leq \|u\|\ \|v\| \ \|\Phi_*\|_{\mathcal B(\mathcal L_1)}, \; \forall \; u,v \in \textrm{span}\{e_{k}\otimes e_{l}\}_{kl}. 
 \end{align} By a standard argument of density, ${\mathcal J}_{\rho}(\Phi_{*})(u,v)$ can be continuosly extended to the whole space $\initsp \otimes \initsp$ and it has associated a bounded operator, denoted without confusion by ${\mathcal J}_{\rho}(\Phi_{*})$, such that 
\begin{align*}
{\mathcal J}_{\rho}(\Phi_{*})(u, v) = \langle u, {\mathcal J}_{\rho}(\Phi_{*}) v\rangle, 
\end{align*} for all $u,v$ in $\initsp \otimes \initsp$. This operator is positive, everywhere defined and bounded on $\initsp\otimes\initsp$. This proves the lemma.
\end{proof}

The next property will become useful when dealing with the action of $\mathcal J_\rho(\Phi_*)$ on general elements of the tensor product $\initsp \otimes \initsp$.

\begin{lemma} \label{continuidad-L1}Let $e_k$ be any fixed element of the orthonormal basis $\{e_i\}_i$ and let $\theta$ be the anti-unitary operator of conjugation w.r.t. that basis. Then, for every $k$, the map from $\initsp$ into $L_1(\initsp)$ defined by $u \longmapsto |e_k\rangle \langle \theta u |$ is continuous. \end{lemma}

\begin{proof} Take any  $u,v\in \initsp$. Simple computations yield

\begin{align*}&\big\| |e_k\rangle \langle \theta u| -|e_k\rangle \langle \theta v| \big\|_1 = \| |e_k\rangle \langle \theta (u-v)| \|_1=\text{tr}\Big(\big| \ |\theta((u-v) \rangle \langle e_k|e_k\rangle \langle \theta(u-v)|\ \big|^\frac{1}{2}\Big)\\
&=\text{tr}\Big( |\theta(u-v) \rangle \langle \theta(u-v) |^\frac{1}{2} \Big) =\text{tr}\Big(\|\theta(u-v)\| \Big | \frac{\theta(u-v)}{\|\theta(u-v)\|} \Big\rangle \Big\langle \frac{\theta(u-v)}{\|\theta(u-v)\|} \Big |  \Big) \\
&=\| \theta(u-v) \|. \end{align*} This proves the result.
\end{proof}
\begin{proposition}\label{jota-rho-simple}For any simple tensor $u\otimes v \in \initsp \otimes\initsp$,  
\begin{eqnarray}\label{identity-J}
\begin{aligned}
\mathcal J_{\rho}(\Phi_{*})u\otimes v&=\sum_i e_i \otimes \Phi_*(\rho^\frac{1}{2}|e_i\rangle \langle \theta u| \rho^\frac{1}{2} )v \\
&=\Big(\sum_{i,j} (\unit \otimes \Phi_*) |e_i\rangle \langle e_j| \otimes \rho^{\frac{1}{2}}|e_i\rangle \langle e_j| \rho^\frac{1}{2}\Big) u\otimes v \\
&=(\unit \otimes \Phi_*) (\omega_\rho) u\otimes v.
\end{aligned}
\end{eqnarray} 	
\end{proposition}

\begin{proof}Let $u\otimes v, u^\prime \otimes v^\prime$ be two arbitrary simple tensors in $\initsp \otimes \initsp$. Let  $u_n = \sum_{1\leq j\leq n} \langle e_{j}, u\rangle e_{j}$, and $v_m,u^\prime_{n^\prime}, v^\prime_{m^\prime}$ denote the respective partial sums of $v, u'$ and $v'$. Then,
\begin{align*} & \big\langle u \otimes v,\mathcal J_{\rho}(\Phi_{*}) u^\prime \otimes v^\prime \big\rangle=\lim_{\substack{n,m\\n^\prime,m^\prime}} \sum_{\substack{k,l\\ k^\prime,l^\prime}} \overline{u}_{k} \overline{v}_{l}u_{k^\prime} v_{l^\prime} \big\langle e_k \otimes e_l ,\mathcal J_{\rho}(\Phi_{*}) e_{k^\prime} \otimes e_{l^\prime} \big\rangle \\
&= \lim_{\substack{n,m\\n^\prime,m^\prime}} \sum_{i,j}\sum_{\substack{k,l\\ k^\prime,l^\prime}} \overline{u}_{k} \overline{v}_{l}u_{k^\prime} v_{l^\prime} \langle e_j,e_{k^\prime} \rangle \langle e_k, e_i \rangle \big\langle e_l, \Phi_*(\rho^\frac{1}{2}|e_i \rangle \langle e_j| \rho^\frac{1}{2}) e_{l^\prime} \big\rangle \\
&=\lim_{\substack{n,m\\n^\prime,m^\prime}} \sum_{i,j} \langle e_j,u^\prime_{n^\prime} \rangle \big \langle u_n \otimes v_m , e_i \otimes \Phi_*(\rho^\frac{1}{2}|e_i \rangle \langle e_j| \rho^\frac{1}{2})v^\prime_{m^\prime} \big\rangle  \label{conv-dom-2}. \end{align*}

An application of Lebesgue's Theorem on Dominated Convergence permits us to  interchange the limits with the infinite sums. The complete positivity of $\Phi_*$, the Cauchy-Schwartz inequality and the well known inequality for a bounded operator $A$ and trace class operator $B$, $\| AB\|_1 \leq \|A\| \ \|B\|_1$, allow us to prove that the general term is dominated by an integrable function of $i,j$. Indeed, 
\begin{align*} & \ \ \ \ |\langle e_j,u^\prime_{n^\prime} \rangle |^2\ | \langle u_n,e_i \rangle |^2\ \Big | \Big\langle v_m, \Phi_*\big(\rho^\frac{1}{2}|e_i \rangle \langle e_j| \rho^\frac{1}{2}\big)v^\prime_{m^\prime} \Big\rangle \Big|^2\\
& =|\langle e_j,u^\prime_{n^\prime} \rangle |^2\ | \langle u_n,e_i \rangle |^2\ \Big| \Big \langle v_m, \sum_l L_l \big| \rho^\frac{1}{2}e_i \big\rangle \big\langle  \rho^\frac{1}{2}e_j \big| L^*_l v^\prime_{m^\prime} \Big \rangle \Big|^2 \\
&= |\langle e_j,u^\prime_{n^\prime} \rangle |^2\ | \langle u_n,e_i \rangle |^2\ \Big|\sum_l \Big \langle L^*_l v_m,\rho^\frac{1}{2}e_i \Big\rangle \Big \langle  \rho^\frac{1}{2}e_j,L^*_l  v^\prime_{m^\prime} \Big  \rangle \Big|^2 \\
&\leq \|u_n\|^2\ \|u^\prime_{n^\prime}\|^2 \Big(\sum_l \big|\big \langle L^*_l v_m, \rho^\frac{1}{2}e_i \big\rangle \big|^2 \Big)\Big(\sum_l \big|\big \langle  \rho^\frac{1}{2}e_j,L^*_l  v^\prime_{m^\prime} \big \rangle \big|^2 \Big)\\
&=\|u_n\|^2\ \|u^\prime_{n^\prime}\|^2 \Big \langle v_m, \Phi_*\big(\rho^\frac{1}{2} |e_i\rangle \langle e_i | \rho^\frac{1}{2} \big) v_m \Big \rangle \Big \langle v^\prime_{m^\prime}, \Phi_*\big(\rho^\frac{1}{2} |e_j\rangle \langle e_j | \rho^\frac{1}{2} \big) v^\prime_{m^\prime} \Big \rangle \\
&=\|u_n\|^2\ \|u^\prime_{n^\prime}\|^2 \text{ tr}\Big(\Phi_*\big(\rho^\frac{1}{2} |e_i\rangle \langle e_i | \rho^\frac{1}{2} \big) |v_m \rangle \langle v_m|  \Big) \text{ tr}\Big(\Phi_*\big(\rho^\frac{1}{2} |e_j\rangle \langle e_j | \rho^\frac{1}{2} \big) |v^\prime_{m^\prime}\rangle \langle v^\prime_{m^\prime}|  \Big) \\
&= \|u_n\|^2\ \|u^\prime_{n^\prime}\|^2\  \Big\|\Phi_*\big(\rho^\frac{1}{2} |e_i\rangle \langle e_i | \rho^\frac{1}{2} \big)|v_m\rangle \langle v_m|\Big\|_1\ \Big\|\Phi_*\big(\rho^\frac{1}{2} |e_j\rangle \langle e_j | \rho^\frac{1}{2} \big) |v^\prime_{m^\prime}\rangle \langle v^\prime_{m^\prime}|\Big\|_1 \\
&\leq \|u_n\|^2\ \|u^\prime_{n^\prime}\|^2\ \||v_m \rangle \langle v_m| \| \ \| |v^\prime_{m^\prime}\rangle \langle v^\prime_{m^\prime}| \| \ \Big\|\Phi_*\big(\rho^\frac{1}{2} |e_i\rangle \langle e_i | \rho^\frac{1}{2} \big)\Big\|_1\ \Big\|\Phi_*\big(\rho^\frac{1}{2} |e_j\rangle \langle e_j | \rho^\frac{1}{2} \big)|\Big\|_1 \\
&\leq \|u\|^2\  \|u^\prime\|^2\ \|v\|^2\ \|v^\prime\|^2\ \|\Phi_*\|^2_{\mathcal B(L_1(\initsp))}\ \  \big\| \rho^\frac{1}{2} |e_i\rangle \langle e_i | \rho^\frac{1}{2}  \big\|_1\ \  \big\|\rho^\frac{1}{2} |e_j\rangle \langle e_j | \rho^\frac{1}{2} \big\|_1 \\
&= \|u\|^2\  \|u^\prime\|^2\ \|v\|^2\ \|v^\prime\|^2\ \|\Phi_*\|^2_{\mathcal B(L_1(\initsp))}\ \ \big\langle e_i, \rho e_i  \big\rangle \ \ \ \  \big\langle e_j, \rho e_j  \big\rangle,
\end{align*} which is clearly integrable as a function of $i,j$.

Taking the limits inside the sum, and recalling that, if $u=\sum_k u_{k} e_k$, then $\theta u = \sum_k \overline{u}_k e_k $,  the following identity holds  \begin{align*}
\langle u \otimes v , {\mathcal J}_{\rho}(\Phi_{*})u^\prime\otimes v^\prime \rangle= \langle u \otimes v, \sum_i e_i \otimes \Phi_*(\rho^\frac{1}{2}|e_i\rangle \langle \theta u^\prime| \rho^\frac{1}{2} )v^\prime \rangle , \end{align*} for any $u,v,u^\prime,v^\prime\in \initsp$. The first equality in (\ref{identity-J}) is obtained by a density argument.

Using the continuity of the map in Lemma \ref{continuidad-L1}, the remaining identities in (\ref{identity-J}) can be straightforwardly derived. Indeed, 
\begin{align*}\mathcal J_{\rho}(\Phi_{*})u\otimes v&=\sum_i e_i \otimes \Phi_*(\rho^\frac{1}{2}|e_i\rangle \langle \theta u| \rho^\frac{1}{2} )v\\
&=\lim_r \sum_i^r e_i \otimes \Phi_*(\rho^\frac{1}{2}|e_i\rangle \langle \theta u| \rho^\frac{1}{2} )v \\
&=\lim_{r,s} \sum_{i,j}^{r,s} (\unit \otimes \Phi_*) e_i \otimes \rho^\frac{1}{2}\big|e_i\big\rangle \big\langle \langle u,e_j \rangle e_j	\big| \rho^\frac{1}{2} v \\
&=\lim_{r,s} \sum_{i,j}^{r,s} (\unit \otimes \Phi_*) |e_i\rangle \langle e_j|u \otimes \rho^\frac{1}{2}|e_i\rangle \langle  e_j	\big| \rho^\frac{1}{2} v \\
&=\lim_{r,s}\Big(\sum_{i,j}^{r,s} (\unit \otimes \Phi_*) |e_i\rangle \langle e_j| \otimes \rho^{\frac{1}{2}}|e_i\rangle \langle e_j| \rho^\frac{1}{2}\Big) u\otimes v \\
&=(\unit \otimes \Phi_*)\Big(\sum_{i,j}  |e_i\rangle \langle e_j| \otimes \rho^{\frac{1}{2}}|e_i\rangle \langle e_j| \rho^\frac{1}{2}\Big) u\otimes v\\
&=(\unit \otimes \Phi_*) (\omega_\rho) u\otimes v. \end{align*}

\end{proof}

\begin{definition} \label{def-jamio}
Let $\rho$ be a state in $\mathcal B (\mathsf h)$, $\{e_i\}_i$ an orthonormal basis of $\mathsf h$ and take $\Phi_* \in  \mathcal{CP}\big(L_1(\initsp)\big)$, the space of all bounded CP maps on $L_{1}(\initsp)$. The $\rho$-Choi-Jamio\l kowski operator of $\Phi_*$ is defined by means of 
\begin{align*} \mathcal J_\rho(\Phi_*)=(\unit \otimes \Phi_*) (\omega_\rho). \end{align*}
\end{definition}

\begin{proposition} If $\text{Im}( \rho^{\frac{1}{2}}) = \initsp$ then the $\rho$-Choi-Jamio\l kowski map is injective.\end{proposition}

\begin{proof}Let $\Phi_*$ and $\Psi_*$ two bounded CP maps on $\mathcal L_1(\mathsf h)$ such that $\mathcal J_\rho(\Phi_*)=\mathcal J_\rho(\Psi_*)$. Let $u\otimes v \in \initsp \otimes \initsp$. Denote the expansion of $u$ w.r.t. the basis $\{e_i\}_i$ by $u=\sum_l u_l e_l$. Then,
\begin{align*}&\big\langle u \otimes v , \mathcal J_\rho(\Phi_*)u \otimes v \big \rangle = \sum_{i,j} \left \langle u \otimes v , |e_i\rangle \langle e_j| \otimes \Phi_*\left(\rho^{\frac{1}{2}} |e_i\rangle \langle e_j| \rho^{\frac{1}{2}}\right) u \otimes v \right\rangle \\
&= \sum_{i,j} \langle u,e_i \rangle \langle e_j ,u \rangle \left\langle v, \Phi\left(\rho^{\frac{1}{2}}| e_i\rangle \langle e_j | \rho^{\frac{1}{2}}\right) v \right\rangle \\
&=\sum_{i,j,l,k} \overline{u}_l u_k \delta_{il} \delta_{kj} \left\langle v, \Phi\left(\rho^{\frac{1}{2}}| e_i\rangle \langle e_j | \rho^{\frac{1}{2}}\right) v \right\rangle \\
&=\sum_{i,j} \left \langle v, \Phi_*\left(\rho^{\frac{1}{2}} | \overline{u}_i e_i \rangle \langle \overline{u}_j e_j|\rho^{\frac{1}{2}} \right) v \right \rangle=\left \langle v, \Phi_*\left( \rho^{\frac{1}{2}}|\theta u \rangle \langle \theta u | \rho^{\frac{1}{2}} \right) v \right \rangle.
\end{align*}
Without loss of generality we can take $\theta u$ instead of $u$, since $\{\theta u : u\in \initsp \}$ remains dense in $\initsp$. Therefore,

\noindent$0=\Big \langle \theta u \otimes v, \big(\mathcal J_\rho(\Phi_*)-\mathcal J_\rho(\Psi_*)\big) \theta u\otimes v \Big\rangle=\Big \langle v, \Big(\Phi_*\big(|\rho^\frac{1}{2}u \rangle \langle \rho^\frac{1}{2}u| \big)- \Psi_*\big(|\rho^\frac{1}{2}u \rangle \langle \rho^\frac{1}{2}u|\big)\Big)v \Big\rangle $. By hypothesis, $Im(\rho^\frac{1}{2}) =\initsp$, thus the set $\{|\rho^\frac{1}{2}u \rangle \langle \rho^\frac{1}{2}u|: u\in \mathsf h \}$ is dense in $\mathcal L_1(\mathsf h)$. So, $\Phi_*=\Psi_*$ coincide on a dense subset of $\mathcal L_1(\initsp)$. Hence we can conclude they are equal.

\end{proof}

\begin{lemma}\label{lemma-change-basis-1} If $\omega'_{\rho}$ is as in Definition \ref{def-jamio} with the orthonormal basis $\{e'_{i}\}_i$ instead $\{e_{i}\}_{i}$, and $U$ is the unitary operator that relates the orthonormal bases $\{e'_{i}\}_i$ and $\{e_{i}\}_i$, i.e.  $Ue_i=e_i^\prime $, then the following relation is satisfied
 \[\omega_\rho^\prime=  (U\theta U^* \theta \otimes \unit ) \omega_\rho (U\theta U^* \theta \otimes \unit )^*, \] where $\omega_{\rho}$ is the state associated with the orthonormal basis $\{e_{i}\}_i$ and $\theta$ is the antiunitary map of conjugation with respect to this basis.
\end{lemma}

\begin{proof} From (\ref{def-omega-rho-simple}), for any simple tensor we have

\begin{align*}
&\Big\langle u\otimes v, \omega_\rho^\prime  u^\prime \otimes v^\prime \Big\rangle 
=\sum_{ij} \langle u\otimes v, U e_i \otimes \rho^\frac{1}{2} U e_i \rangle \langle U e_j \otimes \rho^\frac{1}{2} U e_j, u^\prime \otimes v^\prime \rangle \\
&=\sum_{ij} \langle u, U\theta U^{*}Ue_i \rangle \langle v, \rho^\frac{1}{2} Ue_i \rangle \langle U\theta U^{*} Ue_j, u^\prime \rangle \langle \rho^\frac{1}{2} Ue_j, v^\prime \rangle \\
&=\Big\langle u, U\theta U^{*}\sum_{i}\langle U e_{i}, \rho^{\frac{1}{2}} v\rangle Ue_{i}\Big\rangle \Big\langle U\theta U^{*}\sum_{j}\langle Ue_{j}, \rho^{\frac{1}{2}} v'\rangle Ue_{j}\Big\rangle \\ & = \langle u, U\theta U^{*} \rho^{\frac{1}{2}}v\rangle \langle U\theta U^{*}  \rho^{\frac{1}{2}}v', u'\rangle \\ 
&= \sum_{i,j} \Big\langle u, U\theta U^{*}\theta \langle \rho^{\frac{1}{2}} v, e_{i}\rangle e_{i}\Big\rangle \Big\langle U\theta U^{*}\theta \langle \rho^{\frac{1}{2}} v', e_{j}\rangle e_{j}, u'\Big\rangle \\ 
&= \sum_{i,j} \langle u, U\theta U^{*}\theta e_{i}\rangle \langle v, \rho^{\frac{1}{2}}e_{i}\rangle \langle U\theta U^{*}\theta e_{j}, u' \rangle \langle\rho^{\frac{1}{2}}e_{j}, v' \rangle \\ 
&= \sum_{i,j} \Big\langle u\otimes v, U\theta U^{*} \theta e_{i}\otimes \rho^{\frac{1}{2}} e_{i}\Big\rangle \Big\langle U\theta U^{*} \theta e_{j}\otimes \rho^{\frac{1}{2}} e_{j}, u'\otimes v' \Big\rangle \\ 
&=  \Big\langle u\otimes v, (U\theta U^{*} \theta \otimes \unit) \sum_{i,j} \big\langle  e_{j}\otimes \rho^{\frac{1}{2}} e_{j},(U\theta U^{*} \theta\otimes \unit)^* u'\otimes v' \big\rangle e_{i}\otimes \rho^{\frac{1}{2}} e_{i}\Big\rangle \\ 
&= \Big\langle u\otimes v, (U\theta U^{*}\theta\otimes \unit) \omega_{\rho}(U\theta U^{*}\theta \otimes \unit)^{*} u'\otimes v'\Big\rangle.
	\end{align*}

This identity extends to arbitrary elements of $ \mathsf h\otimes \mathsf h$ by linearity and density. 
\end{proof}

\begin{theorem}\label{change-of-basis-j-states} Let $\{e_{i}\}_i$ and $\{e'_{i}\}_i$ be any two orthonormal bases, $\rho$ a fixed state and $\Phi_*$ a bounded CP map on ${\mathcal L}_{1}(\initsp)$. Then ${\mathcal J}_{\rho}(\Phi_*)$ and ${\mathcal J}'_{\rho}(\Phi_*)$ are related as follows: 
\begin{align}
{\mathcal J}'_{\rho}(\Phi_*) = (U\theta U^{*} \theta\otimes \unit){\mathcal J}_{\rho}(\Phi_*)(U\theta U^{*} \theta \otimes \unit)^{*},
\end{align} where $U$ is the unitary operator which satisfies $e'_{i}=U e_{i}$.
\end{theorem}
\begin{proof}Consider any simple tensor $u\otimes v\in \initsp \otimes \initsp$. By Proposition \ref{jota-rho-simple}, Lemma (\ref{lemma-change-basis-1}) and some computations we get,

\begin{align*} 
{\mathcal J_\rho}^\prime({\Phi_{*}})u\otimes v & =(\unit \otimes \Phi_*) (\omega^{\prime}_\rho) u\otimes v\\
&=(\unit \otimes \Phi_*) ((U\theta U^* \theta \otimes \unit)\omega_\rho (U\theta U^*\theta \otimes \unit)^*) u\otimes v \\ 
&=\lim_{r,s}\Big(\sum_{i,j}^{r,s} (U\theta U^* \theta \otimes \unit)(\unit \otimes \Phi_*) |e_i\rangle \langle e_j| \otimes \rho^{\frac{1}{2}}|e_i\rangle \langle e_j| \rho^\frac{1}{2}(U\theta U^* \theta \otimes \unit)^*\Big) u\otimes v\\
&=(U\theta U^* \theta \otimes \unit)\lim_{r,s}\Big(\sum_{i,j}^{r,s}(\unit \otimes \Phi_*) |e_i\rangle \langle e_j| \otimes \rho^{\frac{1}{2}}|e_i\rangle \langle e_j| \rho^\frac{1}{2}\Big)       (U\theta U^* \theta \otimes \unit)^* u\otimes v \\
& = \big( U\theta U^{*}\theta \otimes \unit\big){\mathcal J}_{\rho}({\Phi_{*}})\big( U\theta U^{*}\theta \otimes \unit\big)^{*} u\otimes v . 
\end{align*} 

Extending this identity to arbitrary elements of $\initsp \otimes \initsp$ by linearity and density the conclusion follows.
\end{proof}

\begin{theorem}\label{J-prop}Let $\Phi_*$ be a CP map. Some properties of $\mathcal J_\rho$ are:
\begin{itemize}

\item[i)] $\mathcal J_\rho(\Phi_*)$ is a positive operator.

\item[ii)] $\mathcal J_\rho(\Phi_*)$ has the following ``\textbf{isometric}" property,
\begin{align}\label{trace-property} 
\|\mathcal J_\rho(\Phi_*)\|_{1}= \tr{\mathcal J_\rho(\Phi_*)}=\tr{\Phi_*(\rho)}= \|\Phi_*(\rho)\|_{1}.
\end{align} Consequently, 
$\mathcal J_\rho(\Phi_*)\text{ is a state} \Leftrightarrow \Phi_*\text{ is a CP trace preserving map.} $ 

\item[iii)] The Choi-Jamio\l kowski transform  $\mathcal J_\rho$ is continuos as a map from the subset $\mathcal{CP}(L_1(\initsp))$, of all completely positive operators on the Banach space $L_1 (\initsp)$, into $L_{1}(\initsp\otimes\initsp)$. That is, the map ${\mathcal J}_{\rho}: {\mathcal C}{\mathcal P}(L_{1}(\initsp)) \rightarrow L_{1}(\initsp\otimes\initsp)$ is continuous, where ${\mathcal C}{\mathcal P}(L_{1}(\initsp))$ is provided with the norm topology of ${\mathcal B}(L_{1}(\initsp))$.

\item[iv)] For every $u, v, z,w \in \mathsf h$ we have 
\begin{align*}\Big\langle u\otimes v, {\mathcal J}_{\rho}(\Phi_{*}) u^\prime\otimes v^\prime\Big\rangle = \left\langle v, \Phi_{*}\big(\rho^{\frac{1}{2}}\theta (|u^\prime\rangle\langle u|)^{*}\theta\rho^{\frac{1}{2}}\big)v^\prime\right\rangle.\end{align*}
\end{itemize}
 \end{theorem}

\begin{proof}\begin{enumerate}

\item[\textit{i)}] Let $u$ be any element in $\initsp \otimes \initsp$ and $\{u_n\}_n$ a sequence in $\text{span}\{e_k \otimes e_l\}\subset \initsp \otimes \initsp$ such that $u_n \to u$.

Since $\langle u , \mathcal J_\rho(\Phi_*)u \rangle =\displaystyle \lim_{n,m} \langle u_n ,\mathcal J_\rho(\Phi_*)u_m \rangle$, it is enough to show that there exists a subsequence of $\{ \langle u_n, \mathcal J_\rho(\Phi_*) u_m \rangle \}_{nm}$ consisting of non-negative elements with limit $\langle u ,\mathcal J_\rho(\Phi_*) u \rangle$. The diagonal subsequence $\{\langle u_n, \mathcal J_\rho(\Phi_*) u_n \rangle\}_n$ converges to $\langle u ,\mathcal J_\rho(\Phi_*) u \rangle$ and, by Lemma $\ref{rho-jamio}$, $\mathcal J(\Phi_*)$ is positive on $\text{span}\{e_k \otimes e_l\}$. This proves $(i)$.

\item[\textit{ii)}]  Due to the above Theorem and the invariance of traces with respect to unitary conjugation, it suffices to prove that (\ref{trace-property}) holds in the case when ${e_i}$ is the basis of $\rho$, i.e., $\rho=\sum_{i}\rho_{i} |e_i \rangle \langle e_i|$. We have that 
\begin{align*}\tr\big({\mathcal J_\rho(\Phi_*)}\big)&=\sum_{i,j} \big \langle e_i \otimes e_j , \mathcal J_\rho(\Phi_*) e_i \otimes e_j \big \rangle =\sum_{i,j} \left\langle e_j , \Phi_*\left(\rho^{\frac{1}{2}} |e_i \rangle \langle e_i | \rho^{\frac{1}{2}} \right)e_j \right\rangle \\
&=\sum_i \tr\big({\Phi_*(\rho^\frac{1}{2} |e_i \rangle \langle e_i | \rho^\frac{1}{2})}\big)= \tr\big({\Phi_*(\sum_{i}\rho_{i} |e_i \rangle \langle e_i |)}\big) \\  &=\tr\big({\Phi_*(\rho)}\big)=\|\Phi_{*}(\rho)\|_{1}.  \end{align*} Since $\Phi_{*}(\rho)$ is positive. 

\item[\textit{iii)}] Since ${\mathcal J}_\rho(\Phi_*)$ is positive, as a direct consequence of $(ii)$, we get 
\begin{align*}\| \mathcal J_\rho(\Phi_*)\|_1=\| \Phi_*(\rho)\|_1 \leq \|\Phi_* \|_{\mathcal B(L_{1}(\initsp))}.
\end{align*} This proves $(iii)$.

\item[\textit{iv)}] By direct computation, for every $u , v, u^\prime,v^\prime  \in \mathsf h$, we have that 
\begin{eqnarray*}
\begin{aligned}
&\Big\langle u  \otimes v  , {\mathcal J}_{\rho}(\Phi_{*}) u^\prime\otimes v^\prime \Big\rangle =\sum_i\left\langle u \otimes v, e_i \otimes \Phi_*(\rho^\frac{1}{2}|e_i \rangle \langle \theta u^\prime| \rho^\frac{1}{2} ) v^\prime \right\rangle \\
&=\sum_i \Big \langle v, \Phi_{*}\Big( \rho^\frac{1}{2}|\langle u , e_{i}\rangle e_{i} \rangle \langle \theta u^\prime|\rho^\frac{1}{2}\Big) v^\prime  \Big \rangle = \Big\langle v, \Phi_{*}\big(\rho^{\frac{1}{2}}\theta(|u^\prime\rangle \langle u |)^{*} \theta\rho^{\frac{1}{2}}\big) v^\prime \Big\rangle.
\end{aligned}
\end{eqnarray*}

\end{enumerate}

 \end{proof}

\section{Quantum Entropy Production Rate}\label{entropy-production}

\subsection{Von Neumann Relative Entropy}

\begin{definition}
The von Neumann relative entropy of two states $\eta$ and $\rho$ is defined as 
\[S(\eta,\sigma)=tr\Big(\eta \log \eta- \eta \log \sigma )\Big)\]
if $\ker(\sigma)\subset \ker(\eta)$ and $\infty$ otherwise.
\end{definition}

\begin{theorem}\label{Nonnega-RE}\textbf{(Non-negativity of the relative entropy)}
\[S(\eta,\rho)\geq 0, \] for all $\eta$, $\rho$. Moreover,  
$S(\eta,\rho)=0$ if and only if $\eta=\rho$.
\end{theorem}

\subsection{The $\Theta$-KMS adjoint QMS}

$\Theta$-SQDB seems to be the most appropriate extension of detailed balance to the non-commutative case;  surprisingly, up to now, a notion of adjoint associated with $\Theta$-SQDB condition has not been discussed. To fulfill this gap we define the $\Theta$-KMS adjoint (or dual) of a given QMS as follows.

\begin{definition}\label{theta-kms-adjoint-def}
Given a reversing operation $\Theta$ and a uniformly continuous QMS ${\mathcal T}=({\mathcal T}_{t})_{t\geq 0}$ on ${\mathcal B}(\initsp)$ with a faithful invariant state $\rho$, we say that ${\mathcal T}$ admits a $\Theta$-KMS adjoint (or dual) QMS with respect to the state $\rho$ if there exists a QMS ${\mathcal T}^{\Theta}=({\mathcal T}^{\Theta}_{t})_{t\geq 0}$ satisfying the $\Theta$-KMS duality relation  
\begin{eqnarray}\label{theta-kms-adjoint}
\tr\Big(\rho^{\frac{1}{2}}\Theta(x^*)\rho^{\frac{1}{2}}{\mathcal T}_{t}(y)\Big)= \tr\Big(\rho^{\frac{1}{2}}\Theta({\mathcal T}^{\Theta}_{t}(x^*)) \rho^{\frac{1}{2}} y \Big), \;  \forall \;  x,y \in {\mathcal B}(\initsp). 
\end{eqnarray} 
\end{definition}

The notion of weighted detailed balance was introduced in \cite{a-f-q}. It is a natural generalization of the quantum detailed balance condition of Frigerio, Kossakowski, Gorini and Verri, \cite{kfgv}, that singles out an interesting class of semigroups that includes, beside QMS with equilibrium states, QMS with non-equilibrium steady states. The generators of these semigroups admit a special GKSL representation in the sense of Parthasarathy, of the form  

\begin{eqnarray}\label{gksl-rep}
{\mathcal L}(x)= i[H, x] - \frac{1}{2} \sum_{k\geq 1} 
\left( L_{k}^{*} L_{k} x - 2 L_{k}^{*} x L_{k} 
+ x L_{k}^{*} L_k \right),
\end{eqnarray} 
where $H, \; L_k \in {\mathcal B}({\initsp})$ with 
$H=H^{*}$ and the series $\sum_{k\geq 1} L_{k}^{*} L_k $ 
is convergent in norm, satisfying the conditions of Theorem 30.16 in \cite{partha}  

\begin{definition}\label{weighted-db}
A uniformly continuous QMS $({\mathcal T}_{t})_{t\geq 0}$, with GKSL generator ${\mathcal L}$ and a faithful invariant normal state $\rho$, is said to satisfy a weighted detailed balance  
condition if ${\mathcal L}$ admits a special GKSL representation and  
there exists a sequence of positive weights $q:=(q_{k})_{k}$ and bounded operators $ {K}, {L}_{k}^{'}$ of 
a (possibly another) special representation of ${\mathcal L}$ 
such that the difference ${\mathcal L}' - {\mathcal L}$ 
has the structure  
\begin{eqnarray}\label{def-weighted-db}
{\mathcal L}' - {\mathcal L} = -2i [K, \cdot] + \Pi, 
\end{eqnarray} where ${\mathcal L}'$ is the KMS adjoint of ${\mathcal L}$, $K=K^{*}$ is bounded and 
\begin{eqnarray}\label{Pi-weighted-db}
\Pi(x)= \sum_{k} (q_{k} -1) {L}_{k}^{'*}x {L}_{k}^{'}. 
\end{eqnarray}
\end{definition}

\begin{proposition} If ${\mathcal T}=({\mathcal T}_{t})_{t\geq 0}$ on ${\mathcal B}(\initsp)$ is a uniformly continuos QMS  with a faithful invariant state $\rho$, whose generator satisfies a weighted detailed balance condition with bounded $K$ and $L_{k}$, then the $\Theta$-KMS adjoint semigroup ${\mathcal T}^{\Theta}$ exists and satisfies  
\begin{eqnarray}\label{theta-kms-duality}
{\mathcal T}^{\Theta} = \Theta \circ{\mathcal T}'\circ\Theta,
\end{eqnarray} with ${\mathcal T}'$ the KMS adjoint QMS. Consequently, ${\mathcal T}^{\Theta}$ is uniformly continuous QMS. 
\end{proposition}
\begin{proof} We have that, 
\begin{eqnarray}
\begin{aligned}
&\tr\Big(\rho^{\frac{1}{2}}{\mathcal T}'\big(\Theta(x^*)\big)\rho^{\frac{1}{2}}y\Big) =\tr\Big(\rho^{\frac{1}{2}}\Theta(x^*)\rho^{\frac{1}{2}}{\mathcal T}_{t}(y)\Big)= \tr\Big(\rho^{\frac{1}{2}}\Theta({\mathcal T}^{\Theta}_{t}(x^*)) \rho^{\frac{1}{2}} y \Big), \\ \;  &\forall \;  x,y \in {\mathcal B}(\initsp).
\end{aligned}
\end{eqnarray} Hence ${\mathcal T}^{\Theta}= \Theta\circ{\mathcal T}' \circ\Theta$. It is well known that ${\mathcal T}'$ is a QMS, \cite{franco3}. Clearly ${\mathcal T}^{\Theta}$ is a uniformly continuous QMS whenever the KMS adjoint QMS ${\mathcal T}'$ is. If the GKSL generator of ${\mathcal T}$ satisfies a weighted detailed balance condition then the GKSL generator of ${\mathcal T}'$ has the structure \[{\mathcal L}'(\cdot) = {\mathcal L}(\cdot) + 2i[K, \cdot] + \sum_{k}(q_{k}-1) L_{k}^{*} \cdot L_{k}.\] With $q_{k}>0$, $K$ and all $L_{k}$ bounded operators. The uniform continuity of ${\mathcal T}$ implies that ${\mathcal L}$ is a bounded operator on ${\mathcal B}(\initsp)$, hence ${\mathcal L}'$ is also bounded as a map from ${\mathcal B}(\initsp)$ into itself. Therefore ${\mathcal T}'$ and, hence, ${\mathcal T}^{\Theta}$ is a uniformly continuous QMS. This finishes the proof.
\end{proof}

\subsection{Deviation from equilibrium}

From now on we consider the class of uniformly continuous QMS ${\mathcal T}$ that admit a uniformly continuous $\Theta$-KMS adjoint QMS ${\mathcal T}^{\Theta}_{t}$. This class include those semigroups satisfying a weighted detailed balance condition. We denote by ${\mathcal T}_{*t}$ and ${\mathcal T}^{\Theta}_{*t}$  the corresponding pre-dual semigroups. 

The relative entropy $S(\mathcal J_\rho(\mathcal T_{*t}),\mathcal J_\rho( \mathcal T^{\Theta}_{*t}))$ is a measure of the deviation from $\Theta$-SQDB of the semigroup $\mathcal T$. We assume that the condition $\textrm{ker}(\mathcal J_\rho( \mathcal T^{\Theta}_{*t}))\subset \textrm{ker}(\mathcal J_\rho(\mathcal T_{*t}))$ holds true for all $t\geq 0$, so that the above relative entropy is finite. Moreover one can define the rate of change of relative entropy as follows.

\begin{definition} 
The Quantum Entropy Production Rate of the uniformly con\-ti\-nuos QMS $\mathcal T_*$, with respect to the invariant state $\rho$, is defined as
\begin{align}\label{entropy-production-rate}
e_p(\mathcal T_*,\rho)= \left.\frac{d}{dt}S(\mathcal J_\rho(\mathcal T_{*t}),\mathcal J_\rho( \mathcal T^{\Theta}_{*t}))\right|_{t=0}. \end{align}
\end{definition}

Notice that in the last definition there is no reference to the orthonormal basis used to compute the $\rho$-Choi-Jamio\l kowski states of $\mathcal T_*$ and $\mathcal T^{\Theta}_*$, this is justified by the following theorem.

\begin{theorem}\label{invariance-entropy}
Let $\{e_{i}\}_i$ be an orthonormal basis of $\mathsf h$, $\Phi_{*}$, $\Psi_{*}$ two CP  trace preserving maps acting on $L_{1}(\initsp)$, and ${\mathcal J}_{\rho}({\Phi_{*}})$, ${\mathcal J}_{\rho}({\Psi_{*}})$ the $\rho$-Choi-Jamio\l kowski states on ${\mathcal B}(\mathsf h \otimes \mathsf h)$, associated with $\Phi_{*}$ and $\Psi_{*}$, respectively. The relative entropy $S\big({\mathcal J}_{\rho}({\Phi_{*}}), {\mathcal J}_{\rho}({\Psi_{*}})\big) $ does not depend on the orthonormal basis $\{e_{i}\}_i$. 
\end{theorem}
\begin{proof}
It suffices to prove that if $\{e'_{i}\}_i$ is another orthonormal basis of $\mathsf h$ and ${\mathcal J}'_{\rho}({\Phi_{*}})$, ${\mathcal J}'_{\rho}({\Psi_{*}})$ are the corresponding states associated with $\Phi_{*}$ and $\Psi_{*}$, then 
\begin{eqnarray}\label{entropy-unitary}
S\big({\mathcal J}'_{\rho}({\Phi_{*}}), {\mathcal J}'_{\rho}({\Psi_{*}})\big)=S\big({\mathcal J}_{\rho}({\Phi_{*}}), {\mathcal J}_{\rho}({\Psi_{*}})\big).
\end{eqnarray} 

Using the properties of the antiunitary operator $\theta$, it follows that $U\theta U^{*}\theta \otimes \unit$ is an unitary operator. Identity (\ref{entropy-unitary}) follows from an application of the well known in\-va\-rian\-ce of relative entropy with respect to unitary conjugations, which is a consequence of its monotonicity with respect to CP maps (Petz-Uhlmann Theorem), and Proposition \ref{change-of-basis-j-states}. 
\end{proof}

As a consequence of the last theorem from now on, in all computations, we can use the orthonormal basis that diagonalizes $\rho$. Hence we can assume that the antiunitary map $\theta$ and the state $\rho$ commute. Indeed, \[\theta\rho u = \theta \sum_{i} \rho_{i}\langle e_{i}, u\rangle e_{i} = \sum_{i} \rho_{i}\langle u, e_{i} \rangle e_{i} = \sum_{i} \rho_{i} \langle e_{i}, \theta u\rangle e_{i} = \rho \theta u,\] for all $u\in\mathsf h$. 

\begin{theorem}\label{deviation-equil-criterion}
Let $({\mathcal T}_{t})_{t\geq 0}$ be a QMS with a faithful invariant state $\rho$ such that $Im(\rho^{\frac{1}{2}})=\initsp$ and $\Theta$-KMS adjoint ${\mathcal T}^{\Theta}_{t}$, the following are equivalent:

\begin{itemize}
\item[(i)] $({\mathcal T}_{t})_{t\geq 0}$ satisfies a $\Theta$-SQDB condition. 

\item[(iii)] The von Neumann relative entropy $S\Big({\mathcal J}_{\rho}({\mathcal T}_{t}), {\mathcal J}_{\rho}({\mathcal T}^{\Theta}_{t})\Big)=0,$ for all $t\geq 0$. 
\end{itemize}
Consequently, the $\Theta$-SQDB condition implies that $e_{p}({\mathcal T}_{*}, \rho)=0$. 
\end{theorem}

\begin{proof} The equivalence follows from Theorem \ref{Nonnega-RE} and the injectiveness of the $\rho$-Choi-Jamio\l kowski map.
\end{proof}

As a consequence of the above theorem, we call \textbf{non-equilibrium steady state} to any invariant state $\rho$ of ${\mathcal T}$ for which $e_p(\mathcal T_*,\rho)\neq 0$.
\begin{definition}
Denote by $\overrightarrow{\omega}_{\rho} (\Phi_*)$ and $\overleftarrow{\omega}_{\rho} (\Phi_*)$ the states (positive functionals) on ${\mathcal B}(\mathsf h \otimes \mathsf h)$ associated with the $\rho$-Jamio\l kowski states ${\mathcal J}_{\rho}(\Phi_*)$ and ${\mathcal J}_{\rho}(\Phi^{\Theta}_*)$ respectively, i.e., for every $x\in{\mathcal B}(\mathsf h \otimes \mathsf h)$ 
\begin{align}
\overrightarrow{\omega}_{\rho}(\Phi_*)(x )  =\tr\Big(\mathcal J_{\rho}(\Phi_*) x\Big), \; \; \textrm{and} \;  \; 
\overleftarrow{\omega}_{\rho}(\Phi_*)(x)  = \tr\Big(\mathcal J_{\rho}(\Phi^{\Theta}_*) x\Big).
\end{align}$\Phi^{\Theta}_*$ denotes the $\Theta$-KMS adjoint of $\Phi_*$ given by (\ref{theta-kms-adjoint}) with $\Phi$ instead ${\mathcal T}_{t}$.
\end{definition}

It is particularly interesting to consider the pair of states in the above definition associated with the QMS $(\mathcal T_t)_{t \geq0}$ and its $\Theta$-KMS adjoint $({\mathcal T}^{\Theta}_{t})_{t \geq0}$ w.r.t. an invariant state $\rho$, \begin{align}
\overrightarrow{\omega}_{\rho}(t)(x )  =\tr\Big({\mathcal J_{\rho}({\mathcal T}_{*t}) x}\Big), \; \; \textrm{and} \;  \; 
\overleftarrow{\omega}_{\rho}(t)(x)  = \tr\Big({\mathcal J_{\rho}({\mathcal T}^{\Theta}_{*t}}) x\Big), 
\end{align}  that we call the forward and the backward state, respectively. As we shall see after, in finite dimension, our forward and backward states as well as its densities ${\mathcal J}_{\rho}(\mathcal T_t)$, and ${\mathcal J}_{\rho}({\mathcal T}^{\Theta}_{t}), \; t\geq0$, respectively, reduce to that introduced by F. Fagnola and R. Rebolledo \cite{fag-reb,fag-reb2}. 

It is worth to stress that the states $\overrightarrow{\omega}_{\rho}(\Phi_*),\overleftarrow{\omega}_{\rho}(\Phi_*)$ are defined on the whole space $\mathcal B(\initsp \otimes \initsp)$. 

\begin{theorem}\label{densities-teo}
For every pair of operators $a, b\in {\mathcal B}(\mathsf h)$ we have that 

\begin{align}\label{state-density-matrix} &\overrightarrow{\omega}_{\rho} (\Phi_*)(a\otimes b) = \tr \big(\rho^{\frac{1}{2}}\theta a^{*} \theta \rho^{\frac{1}{2}}\Phi(b) \big)\\
&\label{reverse-state-density-matrix}\overleftarrow{\omega}_{\rho} (\Phi_*)(a\otimes b) = \tr \big(\rho^{\frac{1}{2}}\theta a^{*} \theta \rho^{\frac{1}{2}}\Phi^{\Theta}(b) \big)\end{align}
\end{theorem}

\begin{proof}
Using \textit{iv)} of Theorem \ref{J-prop} we have for every $u , v , u^\prime , v^\prime  \in \mathsf h$, that  
\begin{eqnarray*}
\begin{aligned}
\overrightarrow{\omega}_{\rho}(\Phi_*)\big(|u^\prime \rangle \langle u |\otimes |v^\prime \rangle \langle v | \big)  &= \tr\big({\mathcal J}_{\rho}(\Phi_*) |u^\prime \otimes v^\prime \rangle \langle u \otimes v |\big) \\ 
& = \langle u \otimes v , {\mathcal J}_{\rho}(\Phi_*)u^\prime \otimes v^\prime \rangle \\
&= \left\langle v , \Phi_* \big( \rho^{\frac{1}{2}}\theta(|u^\prime \rangle \langle u |)\theta\rho^{\frac{1}{2}}\big) v^\prime \right\rangle \\ 
&= \tr\Big( \rho^{\frac{1}{2}}\theta(|u^\prime \rangle \langle u |)^{*}\theta\rho^{\frac{1}{2}}\Phi\big(|v^\prime \rangle \langle v |\big)\Big).
\end{aligned} 
\end{eqnarray*} This identity can be extended to every pair of elements $a, b\in{\mathcal B}(\mathsf h)$, by linearity and density. The proof for $\overleftarrow{\omega}_\rho$ is similar. 
\end{proof}

\begin{proposition}The forward and backward states satisfy the following relation \begin{align} \label{states-form}\overleftarrow{\omega}_{\rho} (\Phi_*)(x)=\overrightarrow{\omega}_\rho(\Phi_*)(\flip x \flip ),\end{align} where $\flip$ is the flip operator on $\mathcal B(\initsp \otimes \initsp)$, defined by $\flip u\otimes v= v \otimes u$.\end{proposition}

\begin{proof}Let $a\otimes b$ any simple tensor on $\mathcal B (\initsp \otimes \initsp)$. Notice that composition of the flip operator with simply tensors acts as follows: $\flip (a\otimes b) \flip=b\otimes a$. Therefore, the definition of $\Phi^\Theta$ together with (\ref{reverse-state-density-matrix}) imply

\begin{align*}&\overleftarrow{\omega}_{\rho} (\Phi_*)(a\otimes b)=\tr{\big(\rho^\frac{1}{2} \theta a^* \theta \rho^\frac{1}{2} \Phi^\Theta (b)\big)}=\tr{\big(( \rho^\frac{1}{2} a^*  \rho^\frac{1}{2}\theta \Phi^\Theta (b)\theta)^*\big)}=\tr{\big(\rho^\frac{1}{2} a \rho^\frac{1}{2} \theta \Phi^\Theta (b)^*\theta \big)}\\
&=\tr{\big(\rho^\frac{1}{2} a \rho^\frac{1}{2}  \Phi^\prime(\theta b^* \theta) \big)}=\tr{\big(\rho^\frac{1}{2}\Phi(a) \rho^\frac{1}{2}  \theta b^* \theta \big)}=\overrightarrow{\omega}_{\rho} (\Phi_*)(b\otimes a)=\overrightarrow{\omega}_{\rho} (\Phi_*)(\flip (a\otimes b) \flip)\end{align*}
The extension to the whole $\mathcal B (\initsp \otimes \initsp)$ is immediate by density and linearity, indeed, 

\begin{align*}&\overleftarrow{\omega}_{\rho} (\Phi_*)(x)=\sum_{i,j,k,l}x_{iljk} \overleftarrow{\omega}_\rho(\Phi_*)(|v_i\otimes v_l \rangle \langle v_j \otimes v_k | )\\
&=\sum_{i,j,k,l}x_{iljk} \overleftarrow{\omega}_\rho(\Phi_*)(|v_i \rangle \langle v_j| \otimes |v_l \rangle \langle v_k|)\\ 
& =\sum_{i,j,k,l}x_{iljk} \overrightarrow{\omega}_\rho(\Phi_*)( \flip |v_i\otimes v_l \rangle \langle v_j \otimes v_k |  \flip)\\
&=\overrightarrow{\omega}_\rho(\Phi_*)(\flip x\flip ). \end{align*}

 \end{proof}

 \begin{lemma}\label{lemma-flip}
 
 \begin{enumerate}
 
 \item The operator $\log ( \flip x \flip)$ is well defined for any strictly positive operator $x$ on $L_{1}(\initsp \otimes \initsp)$ and satisfies
 \begin{align}\log ( \flip x\flip)=\label{log-flip} \flip (\log x)\flip.\end{align}
 
 \item If $\Phi_*$ is a CP map on $L_1 (\initsp)$, then
 \begin{align}  \mathcal J_\rho (\Phi_*^\Theta)=\flip \mathcal J_\rho (\Phi_*) \flip, \label{flip1}\end{align}
 \end{enumerate}
 \end{lemma}

 \begin{proof}\begin{enumerate} \item Being $x$ a compact self-adjoint operator, we can consider its spectral decomposition $x=\sum_{i,j} x_{ij} |u_i \otimes u_j \rangle \langle u_i \otimes u_j |$. By direct computation, \begin{align*} \flip(\log x) \flip &=\sum_{ij}\log x_{ij} \flip  |u_i \otimes u_j \rangle \langle u_i \otimes u_j | \flip\\ 
 &=\sum_{ij}\log x_{ij}  |u_j \otimes u_i \rangle \langle u_j \otimes u_i | \\
 &=\log ( \flip x \flip). \end{align*}
 
\item Take any simple tensor $a\otimes b\in \mathcal B(\initsp \otimes \initsp)$. Since
we have \begin{align*}& \tr{\big(\flip \mathcal J_\rho ( \Phi_*^\Theta)\flip (a\otimes b) \big)}=\overleftarrow{\omega}_\rho (\Phi_*)(b\otimes a)= \tr{ \big( \rho^\frac{1}{2} \theta b^* \theta \rho^\frac{1}{2} \Phi^\Theta(a) \big)}=\tr{\big(\rho^\frac{1}{2} b \rho^\frac{1}{2} \theta \Phi^\Theta(a)^* \theta  \big)}\\
&=\tr{\big( \rho^\frac{1}{2} \Phi^\prime (\theta a^* \theta)\rho^\frac{1}{2} b   \big)}=\tr{\big(  \rho^\frac{1}{2} \theta a^* \theta\rho^\frac{1}{2}  \Phi(b) \big)} =\overrightarrow{\omega}_\rho (\Phi_*)(a \otimes b)=\tr{\big( \mathcal J_\rho (\Phi_*) a\otimes b \big)}, \end{align*} by density and linearity (\ref{flip1}) follows.  
\end{enumerate} 

  \end{proof}

\begin{theorem} The Quantum Relative Entropy of a QMS $(\mathcal T_t)_{t\geq 0}$ with respect to an invariant state $\rho$ satisfies the explicit symmetric formula \begin{align*}S(\mathcal T_{*t},\mathcal T^\Theta_{*t})=\frac{1}{2}\tr\Big(\big(\mathcal J_\rho(\mathcal T_{*t})-\mathcal J_\rho(\mathcal T^\Theta_{*t})\big)\big(\log\mathcal J_\rho(\mathcal T_{*t})-\log\mathcal J_\rho(\mathcal T^\Theta_{*t})\big)\Big).\end{align*}
\end{theorem}

\begin{proof} By (\ref{states-form}) and (\ref{log-flip}) the following equalities hold 
\begin{align*}&S(\mathcal T^\Theta_{*t},\mathcal T_{*t})=\overleftarrow{\omega}_\rho (t)\Big(\log\mathcal J_\rho (\mathcal T^\Theta_{*t})-\log\mathcal J_\rho(\mathcal T_{*t})\Big)\\
&=\overrightarrow{\omega}_\rho(t)\Big(\flip \big(\log\mathcal J_\rho (\mathcal T^\Theta_{*t})-\log\mathcal J_\rho(\mathcal T_{*t})\big)\flip\Big)\\
&=\overrightarrow{\omega}_\rho(t)\Big(\log(\flip \mathcal J_\rho (\mathcal T^\Theta_{*t})\flip)-\log(\flip (\mathcal J_\rho (\mathcal T_{*t})\flip)\Big)\\
&=S(\mathcal T_{*t},\mathcal T^\Theta_{*t})+\overrightarrow{\omega}_\rho(t) \Big(\log(\flip (\mathcal J_\rho (\mathcal T^\Theta_{*t})\flip)-\log\mathcal J_\rho(\mathcal T_{*t})\\
&	\ \ \ +\log\mathcal J_\rho(\mathcal T^\Theta_{*t})-\log(\flip \mathcal J_\rho (\mathcal T_{*t})\flip) \Big)\\
&=S(\mathcal T_{*t},\mathcal T^\Theta_{*t}), \end{align*} where we have used (\ref{flip1}).

The proof is complete recalling that \begin{align*}S(\mathcal T_{*t},\mathcal T^\Theta_{*t})+S(\mathcal T^\Theta_{*t},\mathcal T_{*t})=\tr\Big(\big(\mathcal J_\rho(\mathcal T_{*t})-\mathcal J_\rho(\mathcal T^\Theta_{*t})\big)\big(\log\mathcal J_\rho(\mathcal T_{*t})-\log\mathcal J_\rho(\mathcal T^\Theta_{*t})\big)\Big).\end{align*}
\end{proof}

As an immediate consequence of the above theorem we get an explicit formula for the entropy production rate. 

\begin{corollary}
The quantum entropy production rate is given by 
\begin{align}\label{epr-formula} 
e_p(\mathcal T_{*}, \rho)&=\frac{1}{2}\tr\Big(\big(\mathcal J_\rho(\mathcal L_{*})-\mathcal J_\rho(\mathcal L^\Theta_{*})\big)\lim_{t\to 0^+} \big( \log \mathcal J_\rho (\mathcal T_{*t})-\log \mathcal J_\rho ( {\mathcal T}^\Theta_{*t}) \big)\Big), 
\end{align} where $\mathcal J_\rho(\mathcal L_*)$ denotes $\displaystyle \lim_{t\to 0^+} \frac{\mathcal J_\rho(\mathcal T_{*t})-\mathcal J_\rho(\unit)}{t}$, whenever the limit exists. 
\end{corollary}

\begin{definition} A CP operator $\Phi$ is called parity preserving, with respect to the antiunitary operator $\theta$, if  it commutes with the reversing operation $\Theta(x)=\theta x^* \theta$ $x\in \mathcal B(\initsp)$, i.e., $\Theta(\Phi(a))=\Phi (\Theta(a))$ for all $a\in \mathcal B(\initsp)$. 
 
A QMS $({\mathcal T}_{t})_{t\geq 0}$ is parity preserving if and only if  $\mathcal T_{t}$ is parity preserving for every $t\geq 0$.
\end{definition}
 
\begin{corollary}\label{cor-parity-pres}
If the QMS is parity preserving, the Quantum Relative Entropy satisfies
\begin{align*} S(\mathcal T_{*t},\mathcal T^\Theta_{*t})=S(\mathcal T_{*t},\mathcal T^\prime_{*t}). \end{align*} In other words, the QEPR can be computed using either the usual KMS adjoint or the $\Theta$-KMS adjoint introduced in (\ref{theta-kms-adjoint-def}).
\end{corollary}

\begin{proof}Recall that $\theta^2=\unit$, it is then immediate that $\mathcal T^\prime_{t}=\mathcal T_t^\Theta$ when the QMS is parity preserving. \end{proof}

\subsection{The finite dimensional case}

In finite dimension, F. Fagnola and R. Rebolledo have given a Quantum Entropy Production Rate scheme based on suitable defined forward and backward two-point states motivated by the classical case. These states are defined as follows.

\begin{definition}
Let $\{e_i\}_i$ be the diagonalizing basis of an invariant state $\rho$ of a QMS $\mathcal T$. The forward two-point state is defined on the von Neumann tensor product (i.e., as von Neumann algebras) $\mathcal B(\initsp) \otimes \mathcal B(\initsp)$ by 
\begin{align*} \overrightarrow{\Omega}_t(a\otimes b)=\tr\Big({\rho^{\frac{1}{2}} \theta a^* \theta \rho^{\frac{1}{2}} \mathcal T_t(b) }\Big), \ \ \ \ a,b\in \mathcal B(\initsp);\end{align*}
while the backward two-point state  is
 \begin{align*}\overleftarrow{\Omega}_t(a\otimes b)=\tr\Big({\rho^\frac{1}{2}\theta \mathcal T_t (a^*)\theta \rho^\frac{1}{2} b }\Big), \ \ \ \ a,b,\in\mathcal B(\initsp) ,\end{align*} where $\theta$ is the antiunitary operator of conjugation with respect to the basis $\{e_i\}_i$.
\end{definition}

\begin{theorem}(Fagnola-Rebolledo \cite{fag-reb})
The densities of the forward two-point state $\overrightarrow{\Omega}_t$ are given by $\overrightarrow{D}_t=(\unit \otimes \mathcal T_{*t}) (\omega_\rho)$ and, \textbf{if $\initsp$ is finite dimensional}, the density of $\overleftarrow{\Omega}_t$ is $\overleftarrow{D}_t=(\mathcal T_{*t} \otimes \unit) (\omega_\rho)$, i.e., $\overrightarrow{\Omega}_t(a\otimes b) = \tr \big(\overrightarrow{D}_t a\otimes b\big)$ and $\overleftarrow{\Omega}_t(a\otimes b)= \tr \big(\overleftarrow{D}_t a\otimes b\big)$.
\end{theorem}
Their Entropy Production Rate is defined in terms of these two densities. 

\begin{definition}

\begin{itemize}
\item[(i)] Fagnola-Rebolledo's relative entropy is defined as
\begin{align*}S\Big(\overrightarrow{\Omega}_t,\overleftarrow{\Omega}_t\Big)=\tr\Big(\overrightarrow{D}_t(\log \overrightarrow{D}_t - \log \overleftarrow{D}_t)\Big).\end{align*}

\item[(ii)] The corresponding quantum entropy production rate is
\begin{align}\label{f-r-entropy-pro}
\lim_{t\to 0^+} \frac{S\Big(\overrightarrow{\Omega}_t,\overleftarrow{\Omega}_t\Big)}{t} \end{align}
\end{itemize}
\end{definition}

Items (\ref{state-density-matrix}) and (\ref{reverse-state-density-matrix}) of Theorem \ref{densities-teo} imply that Fagnola-Rebolledo's forward and backward two-point states $\overrightarrow{\Omega}_{t}$ and $\overleftarrow{\Omega}_{t}$ coincide with $\overrightarrow{\omega}_\rho (t)$ and $\overleftarrow{\omega}_\rho (t)$ respectively, on simple tensors.
 The above identities are extended for every element in $\mathcal B(\initsp)\otimes \mathcal B(\initsp) = \mathcal B(\initsp \otimes \initsp)$ using the density of the span of simple tensors in the strong topology.

\begin{theorem}\label{Dt-equals-Jt}
If $\initsp$ is finite dimensional, then the forward and backward densities $\overrightarrow{D}_t$, $\overleftarrow{D}_t$ coincide with our densities $\mathcal J_\rho (\mathcal T_{*t})$, $\mathcal J_\rho (\mathcal T^\Theta_{*t})$, respectively.
\end{theorem}
\begin{proof} Due to (\ref{flip1}), it suffices to prove that if $\initsp$ is finite dimensional, then 
\begin{align} 
(\Phi_* \otimes \unit) (\omega_\rho)=\flip (\unit \otimes \Phi_* )(\omega_\rho) \flip. \label{flip2}
\end{align}

Using an analogous version of (\ref{identity-J}) for $(\Phi_* \otimes \unit)$, we get for simple tensors 
\begin{align*}
\tr{\big( (\Phi_{*}\otimes \unit )(\omega_\rho) a \otimes b \big)} &= 
\sum_{k,l,i,j} \Big\langle e_{k}\otimes e_{l}, \Phi_{*}(|e_{i}\rangle \langle e_{j}|)a e_{k} \otimes \langle e_{j}, \rho^{\frac{1}{2}}b e_{l}\rangle\rho^{\frac{1}{2}} e_{i} \Big\rangle \\ & = 
\sum_{k,l,i,j} \langle e_{j}, \rho^{\frac{1}{2}}b e_{l}\rangle\big\langle e_{k}, \Phi_{*}(|e_{i}\rangle \langle e_{j}|)a e_{k}\big\rangle \langle \rho^{\frac{1}{2}}, e_{i}\rangle \\
&=\sum_{k,l} \Big\langle e_{k}, \Phi_{*}\Big(|\theta\sum_{i}\langle e_{i}, \rho^{\frac{1}{2}} e_{l}\rangle e_{i} \big\rangle\big\langle \theta\sum_{j}\langle e_{j}, \rho^{\frac{1}{2}}b e_{l}\rangle e_{j}| a e_{k}\Big\rangle \\ 
&= \sum_{k,l} \big\langle e_{k}, \Phi_{*}(|\theta\rho^{\frac{1}{2}}e_{l}\rangle \langle \theta \rho^{\frac{1}{2}}b e_{l}|a e_{k}\big\rangle \\ 
&=  \sum_{l} \tr\big(\Phi_{*}(\theta\rho^{\frac{1}{2}}e_{l}\rangle \langle\theta\rho^{\frac{1}{2}}b e_{l}|) a \big) \\ 
&= \tr\big(\rho^{\frac{1}{2}} \theta b^{*}\theta \rho^{\frac{1}{2}}\Phi(a)\big) =  \overrightarrow{\omega}_\rho (\Phi_*)(b\otimes a)\\
&=\tr{\big( \flip(\unit \otimes \Phi_{* })(\omega_\rho)\flip a \otimes b \big)},
\end{align*} where we have used the identity $\sum_{l}|\theta\rho^{\frac{1}{2}}e_{l}\rangle \langle\theta\rho^{\frac{1}{2}}b e_{l}|= \theta\rho^{\frac{1}{2}} b^{*} \rho^{\frac{1}{2}}\theta$, that holds true if and only if $\initsp$ is finite dimensional, as well as the known property (\ref{states-form}) of the states $\overrightarrow{\omega}_\rho,\overleftarrow{\omega}_\rho$. By density and linearity (\ref{flip2}) follows.
\end{proof}

Notice that our approach yields a proof of the fact that the backward state's density is $\mathcal J_\rho (\mathcal T^\Theta_{*t})$ for any initial separable Hilbert space $\initsp$.

\section{Example}\label{example}
Although in our previous work \cite{b-q}, we computed the QEPR for circulant and block circulant QMS, it is worth to compute it again using the formula (\ref{epr-formula}). We consider only block circulant QMS and assume that all vaues of the probability distribution $\alpha:{\mathbb Z}_{p}\times {\mathbb Z}_{q}\mapsto [0,1]$ are non-zero. 

Using the notations and properties in Subsection \ref{circ-qms} we have. 

\begin{lemma} Any circulant QMS is parity preserving. \end{lemma}
\begin{proof}The reversing operation in this context is given by $\Theta(x)= (\theta_p \otimes \theta_q) x^* (\theta_p \otimes \theta_q)$, where $\theta_s$ denotes the conjugation w.r.t. the canonical basis of $\mathbb C^s$, $s=p,q$ and $x\in \mathbb C^p \otimes \mathbb C^q$. By \textit{(i)} of Theorem \ref{prop-circ-qms} and Proposition \ref{theta-proy},
\begin{align*}&(\theta_p\otimes \theta_q) \mathcal T_{*t}(x)^* (\theta_p \otimes \theta_q) =(\theta_p \otimes \theta_q) \sum_{m,n} \Phi_{m,n}(t)(J_p \otimes J_q) x^* (J_p \otimes J_q)^*(\theta_p \otimes \theta_q)\\
&= \sum_{m,n}\Phi_{m,n}(t) (J_p \otimes J_q)  (\theta_p \otimes \theta_q) x^* (\theta_p \otimes \theta_q)(J_p \otimes J_q)^*. \end{align*}  \end{proof}

Since the semigroup is parity preserving, according to Corollary \ref{cor-parity-pres} in the previous section, we can use the KMS adjoint to compute the QEPR.

\begin{theorem}The Quantum Entropy Production Rate of a circulant QMS is
\begin{align*}e_p({\mathcal T_*}, \rho)=\frac{1}{2}  \sum_{m,n}\Big( \alpha(m,n)-\alpha(p-m,q-n)\Big) \, log\frac{\alpha(m,n)}{\alpha(p-m,q-n)}.\end{align*}
\end{theorem} 

\begin{proof} We need to show the existence of the following limit 
\begin{align*}e_p(\mathcal T_{*},\rho)
&= \frac{1}{t}\tr\Big(\big(\mathcal J_\rho(\mathcal L_{*})-\mathcal J_\rho(\mathcal L^\prime_{*})\big)\lim_{t\to 0^+} \big(\log\mathcal J_\rho(\mathcal T_{*t})-\log\mathcal J_\rho(\mathcal T^\prime_{*t})\big)\Big).\end{align*}

For the first factor inside the trace, recall that $\mathcal J_\rho(\mathcal T_{*t})$ and $\mathcal J_\rho(\mathcal T^\prime_{*t})$ are diagonal in the basis $\{u_{mn}\}_{m,n}$ in Theorem \ref{prop-circ-qms}. In this basis, it is immediate that $\mathcal J_\rho (\unit)=|u_{00}\rangle \langle u_{00}|$. If $\mathcal J_\rho (\mathcal L_*):=\displaystyle \lim_{t\to 0^+} \frac{\mathcal J_\rho (\mathcal T_{*t}) - \mathcal J_\rho (\unit)}{t}$ and by $\delta_{mn}$ we denote the Kronecker's delta $\delta_{mn}=1$ if $m=n$ and zero otherwise. We obtain that 
\begin{eqnarray}\label{J-star-L}
\begin{aligned} 
\mathcal J_\rho (\mathcal L_*)&= \lim_{t\to 0^+} \frac{1}{pq} \sum_{m,n} \frac{  \Phi_{m,n} (t) - \delta_{mn}} {t} |u_{mn}\rangle \langle u_{mn}| \\
&=\sum_{m,n\neq 0} \alpha(m,n) |u_{mn} \rangle \langle u_{mn}| - |u_{00} \rangle \langle u_{00}|, 
\end{aligned}
\end{eqnarray} recalling that the limits inside the finite sum have been already computed in \cite{b-q}.  

Since circulant semigroups are parity preserving, Lemma \ref{lemma-flip} implies that 
\[{\mathcal J}_{\rho}({\mathcal T}'_{*t})={\mathcal J}_{\rho}({\mathcal T}^{\Theta}_{*t})=(F\otimes F){\mathcal J}_{\rho}({\mathcal T}_{*t})(F\otimes F),\] therefore

\begin{eqnarray}
\begin{aligned}
{\mathcal J}_{\rho}({\mathcal T}'_{*t}) &= \frac{1}{pq}\sum_{m,n} \Phi_{m,n}(t) (F\otimes F) |u_{mn}\rangle\langle u_{mn}| (F\otimes F) \\ &= \frac{1}{pq} \sum_{m,n} \Phi_{n,m}(t) |u_{mn}\rangle\langle u_{mn}| =  \frac{1}{pq} \sum_{m,n} \Phi_{p-m,q-n}(t) |u_{mn}\rangle\langle u_{mn}|, 
\end{aligned}
\end{eqnarray} where we have used that $(F\otimes F)u_{mn} = u_{nm}$ and $\Phi_{nm}(t)=\Phi_{p-m, q-n}(t)$ are the matrix elements of the adjoint matrix $e^{tQ^{*}}$ of $e^{t Q}$. Now, the limit defining ${\mathcal J}_{\rho}({\mathcal L}')$ is computed as in (\ref{J-star-L}), giving 

\begin{eqnarray}
\begin{aligned}{\mathcal J}_{\rho}({\mathcal L}'_{*})&= \lim_{t\to 0^+} \frac{1}{pq} \sum_{m,n} \frac{  \Phi_{p-m,q-n} (t) - \delta_{mn}} {t} |u_{mn}\rangle \langle u_{mn}| \\
&=\sum_{m,n\neq 0} \alpha(p-m,q-n) |u_{mn} \rangle \langle u_{mn}| - |u_{00} \rangle \langle u_{00}|.
\end{aligned}
\end{eqnarray}
For the second factor, taking into account that $\Phi_{0,0}(t)=\Phi_{p,q}(t)$,
\begin{align*}&\lim_{t\to 0^+} \Big(\log \mathcal J_\rho (\mathcal T_{*t}) -\log \mathcal J_\rho (\mathcal T^\prime_{*t})\Big) \\ &=\lim_{t\to 0^+} \sum_{m,n} \big(\log \Phi_{m,n}(t) - \log \Phi_{p-m,p-n}(t)\big) |u_{mn} \rangle \langle u_{mn}|\\
&=\lim_{t \to 0^+}\sum_{m,n\neq 0} \log \frac{\frac{1}{t}\Phi_{m,n}(t)}{\frac{1}{t}\Phi_{p-m,q-n}(t)} |u_{mn}\rangle \langle u_{mn}| \\ &= \sum_{m,n\neq0}\big( \log \alpha(m,n) - \log\alpha(p-m,q-n) \big) |u_{mn} \rangle \langle u_{mn}|.\end{align*} 

Therefore we have that 

\begin{eqnarray}
\begin{aligned}
& \big(\mathcal J_\rho(\mathcal L_{*})-\mathcal J_\rho(\mathcal L^\prime_{*})\big)\lim_{t\to 0^+} \big(\log\mathcal J_\rho(\mathcal T_{*t})-\log\mathcal J_\rho(\mathcal T^\prime_{*t})\big) \\ &= \sum_{m,n} \big(\alpha(m,n)-\alpha(p-m, q-n)\big)\Big( \log \alpha(m,n) - \log\alpha(p-m,q-n) \Big) |u_{mn} \rangle \langle u_{mn}|.
\end{aligned}
\end{eqnarray} This concludes the proof.

\end{proof}

The reason that makes the last theorem really interesting is the fact that it shows that we can straigthforward call the limits inside (\ref{epr-formula}), $\mathcal J_\rho (\mathcal L_*)$ and $\mathcal J_\rho (\mathcal L^\prime_*)$, for circulant QMS. 

%\section{Acknowledgments}

\end{document}